\documentclass[a4paper, 12pt]{article}
\usepackage[english]{babel}
\usepackage[latin1]{inputenc}
\usepackage{amsfonts,amssymb,epsfig}
\usepackage{color,graphicx,graphics,psfrag}
\usepackage{amsmath,amstext,amscd}
\usepackage{hyperref}
\usepackage{amsthm, cases, subfigure}
\usepackage[percent]{overpic}
\usepackage[titletoc, title]{appendix}

\def\Box{\leavevmode\vbox{\hrule
     \hbox{\vrule\kern4pt\vbox{\kern4pt}%
           \vrule}\hrule}}

\def\paragraph#1{{\bf #1\ }}

\numberwithin{equation}{section}
\newtheorem{lemma}{Lemma}[section]
\newtheorem{theorem}[lemma]{Theorem}
\newtheorem{definition}[lemma]{Definition}
\newtheorem{proposition}[lemma]{Proposition}
\newtheorem{remark}{Remark}[section]

\title{Self-Organized Hydrodynamics with density-dependent velocity}
\author{Pierre Degond$^1$, Silke Henkes$^2$ and Hui Yu$^3$ }
\date{\vspace{-5ex}}
\begin{document}
\maketitle
\medskip
{\footnotesize
 \centerline{1. Department of Mathematics, Imperial College London}
   \centerline{London, SW7 2AZ, United Kingdom}
 \centerline{pdegond@imperial.ac.uk}
}

\medskip

{\footnotesize
 \centerline{2. Institute for Complex Systems and Mathematical Biology, University of Aberdeen}
   \centerline{Aberdeen, AB24 3UE, United Kingdom}
 \centerline{shenkes@abdn.ac.uk}
}

\medskip
{\footnotesize
 \centerline{3. Institut f\"ur Geometrie und Praktische Mathematik, RWTH Aachen University}
   \centerline{Aachen, 52062, Germany}
 \centerline{hyu@igpm.rwth-aachen.de}
}

\bigskip

\centerline{\emph{Dedicated to Peter Markowich to celebrate 30 years of friendship.}}

\begin{abstract}
Motivated by recent experimental and computational results that show a motility-induced clustering transition in self-propelled particle systems,
we study an individual model and its corresponding Self-Organized Hydrodynamic model for collective behaviour that incorporates a density-dependent velocity, as well as inter-particle alignment.
The modal analysis of the hydrodynamic model elucidates the relationship between the stability of the equilibria and the changing velocity,
and the formation of clusters.
We find, in agreement with earlier results for non-aligning particles, that the key criterion for stability is $(\rho v(\rho))'>0$, i.e. a non-rapid decrease of velocity with density.
Numerical simulation for both the individual and hydrodynamic models with a velocity function inspired by experiment demonstrates the validity of the theoretical results.
\end{abstract}

\medskip
\noindent
{\bf Acknowledgements:}
This work has been supported by the Agence Nationale pour la Recherche (ANR) under grant 'MOTIMO' (ANR-11-MONU-009-01), 
by the Engineering and Physical Sciences Research Council (EPSRC) under grant ref. EP/M006883/1, 
and by the National Science Foundation (NSF) under grant RNMS11-07444 (KI-Net). 
P. D. is on leave from CNRS, Institut de Math\'ematiques, Toulouse, France.
He acknowledges support from the Royal Society and the Wolfson foundation through a Royal Society Wolfson Research Merit Award.
H. Y. wishes to acknowledge the hospitality of the Department of Mathematics, Imperial College London, where this research was conducted. 
P. D. and H. Y. wish to thank F. Plourabou\'e (IMFT, Toulouse, France) for enlighting discussions.



\section{Introduction.}
The study of flocking is inspired by the natural behaviour of animal groups, such as flocks of birds and schools of fishes. 
Natural flocks exhibit a range of states, including moving swarms, compact flocks, correlated turning and enhanced density fluctuations. 
To capture flock properties, numerical flocking models such as the Vicsek model \cite{Vicsek_etal_PRL95} have been developed, and subsequently been studied in great detail; see e.g. \cite{Chate_etal_PRE08, Chate_etal_EPJ08, Czirok_Vicsek_PhysicaA00}. 
The Vicsek model exhibits a complex first order transition \cite{Chate_etal_PRE08} between an aligned flocking state at low noise levels and a disordered state at high noise levels, via a band state that depends sensitively on the detailed implementation \cite{Peshkov_etal_PRL12}. 
In parallel, hydrodynamic models of the aligned state have been proposed \cite{Toner_etal_AnnPhys05, Degond_Motsch_M3AS08}, which show that the aligned state exhibits critical scaling fluctuations. 
In particular, enhanced transverse diffusion is responsible for stabilizing true long-ranged order in the Vicsek model, in contrast to the quasi-long-ranged order found in the XY model and cemented in the Mermin-Wagner theorem \cite{Chaikin_Lubensky_95}.

Missing from the Vicsek model are the effects of excluded volume, and repulsion or attraction between individual animals or agents. 
A flurry of recent numerical, analytical and experimental work in the physics community has begun to investigate the effects of non point-like agents. Using soft self-propelled particles, that is a particle model where in addition to short-range repulsive forces, self-propulsion is introduced as a force into fully overdamped Langevin dynamics, several groups showed \cite{Fily_Marchetti_PRL12, Redner_etal_PRL13} that the mix of self-propulsion and volume exclusion has a profound effect on the system properties. These results were first obtained for non-aligning active particles; in this paper we will investigate the effect of additional alignment. 
At intermediate densities, the chief effect of volume exclusion is a slowdown of the effective hydrodynamic velocity $v(\rho)$ where $\rho$ is the density and $v'(\rho)<0$. 
Fits to simulations of self-propelled hard and soft particles and collision-based models suggest a universal form at low and intermediate densities $v(\rho)=v_0(1-c\rho)$, where $c$ depends only on the P\'eclet number $Pe=v_0/ a\nu_r$, where $a$ is the particle radius and $\nu_r$ is the rotational diffusion constant, or more generally the ratio of persistence rate to diffusion rate.

This density-dependent velocity leads to a density instability, and finally to a clustering transition where the system phase-separates into a (single) cluster and a low density gas phase. This motility induced phase separation (MIPS) transition appears to be of a spinodal decomposition type, in a direct analogy to the liquid-gas transition. The transition line is determined by the P\'eclet number, and it is hypothesized that it terminates in a critical point around $Pe\approx 10$ \cite{Redner_etal_PRL13}. At high density, a second transition branch separates the cluster phase from a dense liquid phase, and ultimately a high density, low driving active glassy phase \cite{Fily_etal_SM14}.

Analytically, MIPS was first proposed in a one-dimensional model of interacting run-and-tumble particles \cite{Tailleur_Cates_PRL08}. By mapping the Fokker-Planck equation onto an equivalent equilibrium equation, Tailleur and Cates were able to define an effective free energy with a spinodal transition analogous to the liquid-gas transition. This theoretical approach was later extended to fully brownian particles and tested numerically \cite{Cates_Tailleur_EPL13}.

The effect of alignment on MIPS was first studied by Farrell et al.  \cite{Farrell_etal_PRL12} using a combination of  hydrodynamic equations derived from a microscopic particle model, and direct numerical simulation. The phase diagram contains both homogeneous and MIPS phases, but also travelling bands and finite clusters. 
However, a full understanding of the phase separation mechanism in the presence of alignment is still lacking. 

Instead of the usually constant speed, we introduce a density-dependent velocity $v(\rho)$ to the Vicsek model with alignment between individuals, and study the dynamics of the high density system both through direct simulation and the Self-Organized Hydrodynamic (SOH) formalism.
In this paper, we focus on the deeply aligned phase, and study in detail the location of the instability line and its angular dependence. In the unstable phase, we determine the unstable eigenmode as a function of wave vector and orientation, and determine its growth rate with perturbations. These results are then compared to a numerical solution of the full SOH equations, and a direct solution of the particle model.

This paper is organized as follows: Section \ref{Sec_part} introduces the particle model with the density dependent velocity.
Section \ref{Sec_SOH} presents the derived hydrodynamic model and studies the stability of the inviscid and viscous cases. 
Finally, Section \ref{Sec_Num} presents numerical results from both the particle and hydrodynamic models and ends with a discussion of the growth rate of the instability. 
The appendices detail the derivation of the hydrodynamic model (Sections \ref{app1} and \ref{app2}), and the numerical scheme used to integrate the SOH equations (Section \ref{app_split}).
Finally we detail in Section \ref{viscous_DFT} the discrete Fourier transform used for the numerical evaluation of the growth rate of the instability.

To derive the SOH equations, we employ the mathematical theory developed in \cite{Degond_Motsch_M3AS08} and provide the proof in Appendix \ref{app1} and \ref{app2}. 
One of the key components of the proof is the concept of ``Generalized Collision Invariance'' (GCI) which allows passing to the hydrodynamic limit in spite of the lack of momentum conservation at the particle level.
Further generalization and elaboration of this method can be found in \cite{Frouvelle_M3AS12, Degond_etal_MAA13}.

Our macroscopic model describes two quantities, the density $\rho(\boldsymbol{x},t) \geq 0$ and the mean orientation $\mathbf{\mathbf{\Omega}}(\boldsymbol{x},t)$ 
at position $\boldsymbol{x}\in \mathbb{R}^n$ and time $t\geq0$, and is referred to as the Self-Organized Hydrodynamic (SOH) model:
\begin{subequations}\label{SOH_Nonconst}
\begin{numcases}{}
\partial_t \rho + \nabla_{\boldsymbol{x}}\cdot(c_1v(\rho)\rho\mathbf{\Omega}) = 0,\\
\rho[\partial_t\mathbf{\Omega} + (c_2v(\rho)\mathbf{\Omega}\cdot\nabla_{\boldsymbol{x}})\mathbf{\Omega}] + d\mathcal{P}_{\mathbf{\Omega}^\perp}\nabla_{\boldsymbol{x}}(v(\rho)\rho)
= \gamma\mathcal{P}_{\mathbf{\Omega}^\perp}\Delta_{\boldsymbol{x}}(\rho\mathbf{\Omega}), \\
|\mathbf{\Omega}| = 1,
\end{numcases}
\end{subequations}
where $c_1, c_2$ and $d$ are dimensionless mobility parameters coarse-grained from the ones at the particle level. 
$\gamma$ has the dimensions of a diffusion constant and expresses that polarization diffuses as a result of microscopic alignment interactions.
$v(\rho)$ is the function which specifies the relation between the speed and the density, and therefore non-negative.
Note that the projection operator $\mathcal{P}_{\mathbf{\Omega}^\perp} = \rm{Id} - \mathbf{\Omega}\otimes\mathbf{\Omega}$ will preserve the geometric constraint $|\mathbf{\Omega}| = 1$.
Using a two-dimensional system, we show below that the stability of the equilibrium is related to the behaviour of the mass flux $\rho v(\rho)$.
More precisely, the inviscid model $(\gamma = 0)$ is hyperbolic and hence stable against perturbations if $(\rho v(\rho))'$, the derivative of $\rho v(\rho)$ with respect to $\rho$, is non-negative or certain constraints on the equilibria are satisfied.
The viscous model $(\gamma \neq 0)$ is stable if and only if $(\rho v(\rho))'\geq 0$.

This result shows that a density instability associated to clustering emerges below a threshold mass flux derivative and also that we obtain steady states with constant density for mass fluxes above this threshold. Our result agrees with the conclusions of \cite{Tailleur_Cates_PRL08,Fily_Marchetti_PRL12, Redner_etal_PRL13,Cates_Tailleur_EPL13, Fily_etal_SM14} for the soft particle and hydrodynamics models where self-propulsion decreases with local density. In common with Farrell et al. \cite{Farrell_etal_PRL12}, we find that the alignment and clustering transitions are largely independent of each other; however we find correlations between orientation of a perturbation and stability which have not previously been explored.
This analysis on the SOH model is supported by numerical simulations with forms of $v(\rho)$ inspired by experiment for both the individual-based model and the SOH model.
In particular, in the high concentration limit, we can decompose the instability as a sum of unstable eigenmodes, the growth rate of which we then study numerically.
Given a steady state for the density, we observe a positive correlation between the growth rate and the stationary polarisation orientation for both the particle and hydrodynamic models.

\section{Particle model with density-dependent velocity.}\label{Sec_part}

In this section, we introduce the particle model which is the starting point for the derivation of the SOH model (\ref{SOH_Nonconst}).
Consider a system of $N$ self-propelled particles in $\mathbb{R}^n$. 
Let $t$ be the time, $\mathbf X_i(t)$ the position of the $i$-th particle and $\mathbf{\omega}_i(t)$ its velocity orientation. 
Then the time evolution for the $i$-th particle is given by
\begin{subequations}\label{IBM}
\begin{numcases}{}
\frac{d \mathbf{X}_i}{dt} = v(m_i)\mathbf{\omega}_i,\\
d\mathbf{\omega}_i = \mathcal{P}_{\mathbf{\omega}_i^{\perp}}(\nu\bar{\mathbf{\omega}}_i dt+ \sqrt{2D}dB_t^i),\\
|\mathbf{\omega}_i| = 1. \label{IBMc}
\end{numcases}
\end{subequations}
Here $\mathcal{P}_{\mathbf{\omega}_i^{\perp}}$ represents the projection on the plane which is perpendicular to $\mathbf{\omega}_i$, which allows the geometrical constraint on $\mathbf{\omega}_i$ (\ref{IBMc}) to be preserved. 
$B_t^i$ is a Brownian motion with noise strength $D$ and 
$\bar{\mathbf{\omega}}_i$ is the mean velocity oriention in the neighborhood $\mathcal{B}_{R_1}(\mathbf X_i) = \{ \mathbf X_j: |\mathbf X_i-\mathbf X_j| \leq R_1\}$, where $R_1$ is the range of alignment. 
More precisely, we have
\[
\bar{\mathbf{\omega}}_i = \frac{\mathcal{J}_i}{|\mathcal{J}_i|} \quad \text{ with } \mathcal{J}_i = \sum_{j=1}^NK_1(\mathbf X_i- \mathbf X_j)\mathbf{\omega}_j,
\]
where $K_1( \mathbf X)$ is the kernel for the alignment with a compact support in $\mathcal{B}_{R_1}$, and the constant parameter $\nu$ is the alignment rate.

The new ingredient of this individual-based model compared to the Vicsek model is the dependency of the velocity on the mass in the neighborhood, 
described by the function $v(m_i)$ where $m_i$ is the density of the neighbourhood of $\mathbf X_i$.
Let $R_2$ be the interaction range and $|\mathcal{B}_{R_2}|$ the volume of the ball $\mathcal{B}_{R_2}(\mathbf X_i) = \{ \mathbf X_j: |\mathbf X_i-\mathbf X_j| \leq R_2\}$.
Then $m_i$ is defined as
\begin{equation}
m_i = \frac{1}{|\mathcal{B}_{R_2}|} \sum_{j=1}^N K_2(\mathbf X_i-\mathbf X_j),
\end{equation}
where $K_2(\mathbf X)$ is a kernel defined on a compact support in $\mathcal{B}_{R_2}$. 
For instance, $K_2$ can be chosen as $\chi_{\mathbf X_i}(|\mathbf X_i - \mathbf X|)$, the characteristic function on $\mathcal{B}_{R_2}(\mathbf X_i)$. 

An intuitive motivation for the density-dependent velocity is as follows: Particles that are not point-like will experience collisions. While in thermal systems this just randomises directions, for active, persistent motion the combination of collision and persistence just slows the particle down, much like a crowd slows a pedestrian who passes through. A simple scaling argument shows that the collision rate is proportional to density, leading to $v(\rho) \approx v_0(1-c \rho)$. 
For aligning systems, like here, the situation is a bit more complex since colliding particles will also align, eventually removing many of the collisions if the system has strong overall polarisation. Here, we keep a generic form of $v(\rho)$, and investigate the onset of instability as a function of wave vector and angle with the polarisation direction.

\section{Stability analysis of the SOH model with density-dependent velocity.}\label{Sec_SOH}

The derivation of the SOH model (\ref{SOH_Nonconst}) is analogous with the work in \cite{Degond_Motsch_M3AS08} and \cite{Degond_etal_CMS15}; for the details please see Appendix \ref{SOH_derivation}. 

We derive the stability criteria for the inviscid and viscous cases in a two-dimensional space in the following subsections.
Consider the SOH model for $\boldsymbol{x} = (x,y) \in \mathbb{R}^2$. 
Since $|\mathbf{\Omega}| = 1$, we define $\mathbf{\Omega}(x,y,t) = (\cos\theta(x,y,t), \sin\theta(x,y,t))$ through the angle function $\theta(x,y,t)$ of the vector $\mathbf{\Omega}$.
Then the projection operator $\mathcal{P}_{\mathbf{\Omega}^\perp}$ becomes a matrix operator 
$\left(\begin{array}{cc} \sin^2\theta & -\sin\theta\cos\theta \\ -\sin\theta\cos\theta &\cos^2\theta\end{array}\right)$.
Without loss of generalization, we scale the system (\ref{SOH_Nonconst}) such that $c_1 = 1$, and arrive at a system of $\rho$ and $\theta$:
\begin{align}\label{viscous_nonlinear}
&\partial_t
\left(\begin{array}{c}
\rho\\\theta\end{array}\right)
+ A_x\partial_x\left(\begin{array}{c}
\rho\\\theta\end{array}\right)
+A_y\partial_y\left(\begin{array}{c}
\rho\\\theta\end{array}\right) \notag\\
= &\left(\begin{array}{c} 
0 \\ \gamma(\partial^2_x\theta + 2\frac{\partial_x\rho}{\rho}\partial_x\theta + \partial^2_y\theta + 2\frac{\partial_y\rho}{\rho}\partial_y\theta)
\end{array}\right),
\end{align}
where the two matrices $A_x(\rho,\theta)$ and $A_y(\rho,\theta)$ are given by
\begin{align*}
&A_x(\rho,\theta) = \left(\begin{array}{cc}
\tilde{v}'(\rho)\cos\theta & -\tilde{v}(\rho)\sin\theta\\
-d\frac{\tilde{v}'(\rho)}{\rho}\sin\theta & c_2\frac{\tilde{v}(\rho)}{\rho}\cos\theta\end{array}\right), \\
&A_y(\rho,\theta) = \left(\begin{array}{cc}
\tilde{v}'(\rho)\sin\theta & \tilde{v}(\rho)\cos\theta\\
d\frac{\tilde{v}'(\rho)}{\rho}\cos\theta & c_2\frac{\tilde{v}(\rho)}{\rho}\sin\theta\end{array}\right),
\end{align*}
and $\tilde{v}(\rho) = \rho v(\rho)$ and $\tilde{v}'(\rho)$ is the derivative of $\tilde{v}$ with respect to $\rho$.

To simplify, suppose that $(\rho,\theta)$ depends only on $x$, i.e., 
we are interested in the propagation of waves with arbitrary orientation $\mathbf{\Omega}$ in the horizontal direction.
Let $(\rho_s(x), \theta_s(x))$ denote the equilibrium solutions and they must satisfy the following system:
\begin{subequations}
\begin{numcases}{}
\partial_x(\tilde{v}(\rho_s)\cos\theta_s ) = 0,\\
\left[\left(\frac{c_2}{d} + 1\right)\cos^2\theta_s - 1\right]\tilde{v}(\rho_s)\partial_x\theta_s 
= \frac{\gamma}{d\rho_s}\cos\theta_s\partial_x(\rho_s^2\partial_x\theta_s).
\end{numcases}
\end{subequations}
Let us then expand around them with a small perturbation parameter $\sigma$:
\[
\rho(x,t) = \rho_s(x) + \sigma\rho_{\sigma}(x,t), \quad
\theta(x,t) = \theta_s(x) + \sigma\theta_{\sigma}(x,t).
\]
Dropping the higher order terms $\mathcal{O}(\sigma^2)$ and using $(\rho_\sigma, \theta_\sigma)$ to represent the first order perturbation, we arrive at the system below 
\begin{equation}\label{linear_viscous}
\partial_t
\left(\begin{array}{c}
\rho_\sigma\\\theta_\sigma\end{array}\right)
+ A_x(\rho_s, \theta_s)\partial_x\left(\begin{array}{c}
\rho_\sigma\\\theta_\sigma\end{array}\right)
= \left(\begin{array}{c} 
0 \\ \gamma\partial^2_x\theta_\sigma
\end{array}\right).
\end{equation}

\subsection{The inviscid model, i.e., $\gamma = 0$.}

The condition for the linearised system \eqref{linear_viscous} with $\gamma = 0$ to be hyperbolic, i.e. with perturbations that decay back to stability, is
\begin{theorem}
The system \eqref{linear_viscous} with $\gamma = 0$ is hyperbolic if $\tilde{v}'(\rho_s) \geq 0$ or 
\[
\tilde{v}'(\rho_s) < 0 \quad \text{ and } \quad \frac{\left(\tilde{v}'(\rho_s) - c_2\frac{\tilde{v}(\rho_s)}{\rho_s}\right)^2}{-4d\frac{\tilde{v}'(\rho_s)\tilde{v}(\rho_s)}{\rho_s}} \geq \tan^2\theta_s.
\] 
\end{theorem}
\begin{proof} We examine the hyperbolicity of the above system by looking at the eigenvalues $\lambda$ of the matrix $A_x(\rho_s, \theta_s)$. 
Neglecting the subscript $s$ of $(\rho_s, \theta_s)$, the equation $|A_x-\lambda{\rm Id}| = 0$ gives
\begin{equation}
\lambda^2 - \left(\tilde{v}'(\rho) + c_2\frac{\tilde{v}(\rho)}{\rho}\right)\cos\theta\lambda
+ \frac{\tilde{v}'(\rho)\tilde{v}(\rho)}{\rho}(c_2\cos^2\theta-d\sin^2\theta) = 0,
\end{equation}
where the discriminant 
\begin{equation}
\Delta_x = \left(\tilde{v}'(\rho) - c_2\frac{\tilde{v}(\rho)}{\rho}\right)^2\cos^2\theta
+ 4\frac{\tilde{v}'(\rho)\tilde{v}(\rho)}{\rho}d\sin^2\theta.
\end{equation}
If $\tilde{v}'(\rho)\geq 0$, in other words, if $\rho v(\rho)$ is an nondecreasing function, 
the system (\ref{linear_viscous}) with $\gamma = 0$ possesses two real eigenvalues and is hence always hyperbolic. 
Otherwise we obtain two real eigenvalues if 
\begin{align}\label{inviscid_Delta_neg}
\Delta_x \geq 0 \quad \Longleftrightarrow \quad
\frac{\left(\tilde{v}'(\rho) - c_2\frac{\tilde{v}(\rho)}{\rho}\right)^2}{-4d\frac{\tilde{v}'(\rho)\tilde{v}(\rho)}{\rho}} \geq \tan^2\theta.
\end{align} 
\end{proof}

\begin{remark} 
An analogous relation can be derived for a perturbation in only the $y$-direction, where the right-hand side of the inequality \eqref{inviscid_Delta_neg} is replaced by $\cot^2 \theta$. The conclusion here is that the result above is generic, i.e. independent of the direction of perturbation, if we define $\theta$ as the angle between the direction of perturbation and the direction of $\mathbf{\Omega}$.
\end{remark}

\subsection{The viscous case, i.e., $\gamma \neq 0$.}\label{Sec_Viscous}

We have the following result:
\begin{theorem}\label{thm_viscous_stability}
The system \eqref{linear_viscous} is stable around the zero solutions if $\tilde{v}'(\rho_s)\geq 0$ and unstable otherwise.
\end{theorem}
\begin{proof}
We apply a Fourier transform from position variable $x$ to wave number $\xi$ to $(\rho_\sigma, \theta_\sigma)$:
\begin{align*}
\rho_\sigma(x,t) = \frac{1}{2\pi}\int_{-\infty}^\infty\hat\rho_\sigma(\xi,t) e^{i\xi x}\,d\xi
\quad\text{ with }\quad \hat\rho_\sigma(\xi,t) = \int_{-\infty}^\infty\rho_\sigma(x,t)e^{-i\xi x}\,dx,\\
\theta_\sigma(x,t) = \frac{1}{2\pi}\int_{-\infty}^\infty\hat\theta_\sigma(\xi,t) e^{i\xi x}\,d\xi
\quad\text{ with }\quad \hat\theta_\sigma(\xi,t) = \int_{-\infty}^\infty\theta_\sigma(x,t)e^{-i\xi x}\,dx.
\end{align*}
The system of $(\hat{\rho_\sigma}, \hat{\theta_\sigma})$ in a matrix form is
\begin{align}\label{sys_FT}
\partial_t\left(\begin{array}{c}\hat\rho_\sigma\\\hat\theta_\sigma\end{array}\right)
+i\xi A_\xi(\rho_s,\theta_s) \left(\begin{array}{c}\hat\rho_\sigma\\\hat\theta_\sigma\end{array}\right) = \boldsymbol{0},
\end{align}
where
\begin{align*}
A_\xi(\rho_s,\theta_s) = \left(\begin{array}{cc}
\tilde{v}'(\rho_s)\cos\theta_s & -\tilde{v}(\rho_s)\sin\theta_s \\
-d\frac{\tilde{v}'(\rho_s)}{\rho_s}\sin\theta_s &
- i\gamma\xi + c_2\frac{\tilde{v}(\rho_s)}{\rho_s}\cos\theta_s
\end{array}\right).
\end{align*}
The stability of the system (\ref{linear_viscous}) is equivalent to requiring that ${\rm Im}\lambda$, the imaginary part of the eigenvalues of $A_\xi$, and $\xi$ have the opposite signs, so that the decay constant of the perturbation is negative. 
Note that $|A_\xi - \lambda {\rm Id}| = 0$ gives
\begin{align}
0=\lambda^2 &- \left[\Big(\tilde{v}'(\rho_s) + c_2\frac{\tilde{v}(\rho_s)}{\rho_s}\Big)\cos\theta_s
- i\gamma\xi\right]\lambda \nonumber\\
 &+ \tilde{v}'(\rho_s)\cos\theta_s\left(-i\gamma\xi + c_2\frac{\tilde{v}(\rho_s)}{\rho_s}\cos\theta_s\right)
 - d\frac{\tilde{v}(\rho_s)\tilde{v}'(\rho_s)}{\rho_s}\sin^2\theta_s.
\end{align}
The discriminant is
\begin{align*}
\Delta = \left[\Big(\tilde{v}'(\rho_s) - c_2\frac{\tilde{v}(\rho_s)}{\rho_s}\Big)\cos\theta_s
+ i\gamma\xi\right]^2 
+ 4d \frac{\tilde{v}(\rho_s)\tilde{v}'(\rho_s)}{\rho_s}\sin^2\theta_s. 
\end{align*}
and the eigenvalues $\lambda$ can be written as 
\begin{align}\label{eq_lambda}
\lambda 
&=\frac12\left[\Big(\tilde{v}'(\rho_s) + c_2\frac{\tilde{v}(\rho_s)}{\rho_s}\Big)\cos\theta_s \pm {\rm Re}\sqrt{\Delta}
+ i(\pm{\rm Im}\sqrt{\Delta} - \gamma\xi)\right],
\end{align}
where $\sqrt{\Delta}$ denotes the square root of the complex number $\Delta$. 

Next we examine the sign of ${\rm Im}\lambda$ by comparing the values $|{\rm Im}\sqrt{\Delta}|$ and $|\gamma\xi|$ since
\[
{\rm Im}\lambda = {\rm Sign}(\xi) \big(\pm |{\rm Im}\sqrt{\Delta}| - |\gamma\xi|\big),
\]
where ${\rm Sign}(\cdot)$ is the sign function for any real number.
We introduce several notations:
\[
a = \left(\tilde{v}'(\rho_s) - c_2\frac{\tilde{v}(\rho_s)}{\rho_s}\right)\cos\theta_s, \;
b = \gamma\xi, \; e = 4d \frac{\tilde{v}(\rho_s)\tilde{v}'(\rho_s)}{\rho_s}\sin^2\theta_s,
\text{ and } \sqrt{\Delta} = \alpha + i\beta.
\]
The discriminant can be written as $\Delta = (a + ib)^2 + e$, and
${\rm Im}\lambda = {\rm Sign}(\xi) (\pm |\beta| - |b|)$. 
If $|b| \geq |\beta|$, then ${\rm Im}\lambda$ alway has the opposite sign to $\xi$ and the system is stable. 
Otherwise, it is instable.

Using the expressions for $\Delta$ and $\sqrt{\Delta}$, we obtain the equalities for the real and imaginary parts of $\Delta$:
\[
a^2 - b^2 + e = \alpha^2 - \beta^2 \quad \text{ and } \quad ab = \alpha\beta.
\]
There exist three cases:
\begin{enumerate}
\item[(i)] $\beta = 0$. Then $\sqrt{\Delta} = \alpha$. Hence the system possesses two real eigenvalues and is alway stable.
\item[(ii)] $\beta \neq 0$ and $a = 0$. It follows that $\alpha = 0$ and $b^2 - e = \beta^2$. 
Moreover, $a=0$ implies that $\tilde{v}'(\rho_s) - c_2\frac{\tilde{v}(\rho_s)}{\rho_s} = 0$ or $\cos\theta_s = 0$. 
In either case, we have 
\[
\tilde{v}'(\rho) \geq 0 \Longrightarrow e \geq 0 \Longrightarrow |b| \geq |\beta|,
\]
and
\[
\tilde{v}'(\rho) < 0 \Longrightarrow e < 0 \Longrightarrow |b| < |\beta|,
\]
\item[(iii)] $\beta \neq 0$ and $a \neq 0$. We have
\[
\left|\frac{b}{\beta}\right| = \left|\frac{\alpha}{a}\right|
\quad \text{ and } \quad 
e = \beta^2\Big(\left|\frac{b}{\beta}\right|^2 - 1\Big) + a^2\Big(\left|\frac{\alpha}{a}\right|^2-1\Big) 
= (\beta^2+a^2)\Big(\left|\frac{b}{\beta}\right|^2 - 1\Big).
\]
If $\tilde{v}'(\rho) \geq 0$, then $e \geq 0$ and $|b| \geq |\beta|$. 
If $\tilde{v}'(\rho) < 0$, then $e < 0$ and $|b| < |\beta|$.
\end{enumerate}
To summarize, if $\tilde{v}'(\rho) \geq 0$, $\pm |\beta| - |b|\leq 0$ for any $\xi$ and ${\rm Im}\lambda$ has a different sign from $\xi$.
If $\tilde{v}'(\rho) < 0$, there is always one eigenvalue $\lambda$ that has the same sign as $\xi$ and sustains the unstable eigenmode, and hence the overall system will be unstable.
\end{proof}

\begin{remark}
Usually, diffusive models are stable for large values of $\xi$ because diffusion becomes stronger. Note that this is not the case here.
Indeed, when the system is unstable, i.e., $\tilde v'(\rho_s) < 0$, 
the unstable mode associated with the eigenvalue $\lambda$ (where $\xi{\rm Im}\lambda>0$) grows with an exponential rate $\xi{\rm Im}\lambda$. 
Moreover,
\begin{equation}\label{growth_rate_infty}
\lim_{\xi \to \infty} \xi{\rm Im}\lambda = - \frac{d\tilde v(\rho_s)\tilde v'(\rho_s)\sin^2\theta_s}{\gamma\rho_s}.
\end{equation}
In order to prove the above limit as $\xi \to \infty$, we study the Taylor expansion of $\sqrt{\Delta}$ around $\frac{1}{\xi}$.
We employ the notations introduced in the proof of Theorem \ref{thm_viscous_stability} and write $\Delta$ in the polar representation:
\[
\Delta = [(a^2 + e -\gamma^2\xi^2)^2 + 4a^2\gamma^2\xi^2]^{\frac12} e^{i\theta_\Delta},
\]
where
\[
\theta_\Delta = \arctan\left(\frac{2a\gamma\xi}{a^2+e-\gamma^2\xi^2}\right) + \pi.
\]
Then
\begin{align*}
{\rm Re}\sqrt{\Delta} &= \pm[(a^2 + e -\gamma^2\xi^2)^2 + 4a^2\gamma^2\xi^2]^{\frac14}\cos\left(\frac{\theta_\Delta}{2}\right),\\
{\rm Im}\sqrt{\Delta} &= \pm[(a^2 + e -\gamma^2\xi^2)^2 + 4a^2\gamma^2\xi^2]^{\frac14}\sin\left(\frac{\theta_\Delta}{2}\right).
\end{align*}
We are only interested in the imaginary part as $\xi \to \infty$. The Taylor expansion at $\frac{1}{\xi}$ gives
\begin{align*}
{\rm Im}\sqrt{\Delta} &= \pm\xi\left[\gamma + \frac{a^2-e}{2\gamma}\left(\frac{1}{\xi}\right)^2 + \mathcal O\left(\frac{1}{\xi^3}\right)\right]
\left[1 - \frac{a^2}{2\gamma^2}\left(\frac{1}{\xi}\right)^2 + \mathcal O\left(\frac{1}{\xi^3}\right)\right] \\
&= \pm\left[\gamma\xi - \frac{e}{2\gamma}\frac{1}{\xi} + \mathcal O\left(\frac{1}{\xi^2}\right)\right].
\end{align*}
For the unstable mode, we have
\[
\xi{\rm Im}\lambda = -\frac{e}{4\gamma} + \mathcal O\left(\frac{1}{\xi}\right).
\]
In addition, Eq. \eqref{growth_rate_infty} tells us that in the limit we already took, $\xi{\rm Im}\lambda$ is increasing with respect to $\theta_s$ for $\rho_s$ fixed. 

As $\xi \to 0$, note that ${\rm Im}\sqrt{\Delta}$ is approaching zero as well, which indicates a slow growth rate of the unstable modes for the long wavelength.
\end{remark}

\section{Numerical results of the microscopic and macroscopic models.}\label{Sec_Num}

The particle model is solved using the circle method that updates the angle of $\omega_i$ at each discrete time step; the reader can refer to \cite{Motsch_Navoret_MMS11} for details.
The challenge in the numerical resolution of the SOH model is caused by the geometric constraint $|\mathbf{\Omega}| = 1$ and hence the non-conservative nature of the system (\ref{SOH_Nonconst}).
The splitting scheme proposed in \cite{Motsch_Navoret_MMS11} solves the relaxation problem of (\ref{SOH_Nonconst}) in the sense that the norm of $\mathbf{\Omega}$ is initially not restricted to be 1, but then takes the zero limit of an expansion parameter $\eta$ which realizes the geometric constraint on $\mathbf{\Omega}$. 
We will extend this idea to the SOH model, starting with the following proposition of the relaxation model.
\begin{proposition}
Let $\eta$ be a scalar parameter and $(\rho^\eta, \mathbf{\Omega}^\eta)$ the solutions to the relaxation model
\begin{subequations}\label{SOH_relax}
\begin{numcases}{}
\partial_t \rho^\eta + \nabla_{\boldsymbol{x}}\cdot(c_1v(\rho^\eta)\rho^\eta\mathbf{\Omega}^\eta) = 0,\\
\partial_t(\rho^\eta\mathbf{\Omega}^\eta) + \nabla_{\boldsymbol{x}}\cdot(c_2v(\rho^\eta)\rho^\eta\mathbf{\Omega}^\eta\otimes\mathbf{\Omega}^\eta) 
+ d\nabla_{\boldsymbol{x}}(v(\rho^\eta)\rho^\eta) \nonumber \\
\hspace{3cm} 
 - \gamma\Delta_{\boldsymbol{x}}(\rho^\eta\mathbf{\Omega}^\eta)
= \frac{\rho^\eta}{\eta}(1-|\mathbf{\Omega}^\eta|^2)\mathbf{\Omega}^\eta. \label{SOH_relax_b}
\end{numcases}
\end{subequations}
Then $(\rho^\eta, \mathbf{\Omega}^\eta)$ converges to the solutions of the SOH model (\ref{SOH_Nonconst}) as $\eta$ goes to zero.
\end{proposition}
\begin{proof}
The main idea is based on the fact that the right-hand side of (\ref{SOH_relax_b}) is parallel to $\mathbf{\Omega}^\eta$. 
Then the proof is analogous to the one in \cite{Motsch_Navoret_MMS11}.
\end{proof}

The numerical method, the so-called splitting scheme, is based on the work in \cite{Motsch_Navoret_MMS11} and details are provided in Appendix \ref{app_split}.

We test the particle model and the SOH model on a rectangular domain $[0,L_x]\times[0,L_y]$ with $L_x = L_y = 10$
and impose periodic boundary conditions in both directions.
We choose a time step $\Delta t = 0.001$ to discretize time. 
For the function $v(\rho)$, we consider a monotonically decaying power law with exponent $\alpha$ inspired by experimental and numerical results, and with an lower density threshold $\rho^*$: below $\rho^*$, the velocity is essentially constant, and then it rapidly decreases with power law $\alpha$,
\[
v(\rho) = \beta\left(\frac{\rho}{\rho^*} + 1\right)^{-\alpha} \text{ with three positive parameters } \rho^*, \alpha, \beta. 
\]
Note that 
\[
\tilde{v}'(\rho) = \frac{\rho^* + (1-\alpha)\rho}{\rho + \rho^*}v(\rho).
\]
The parameters $(\rho^*, \alpha, \beta)$ will be carefully chosen to place the system in the different regimes of stability derived in Section \ref{Sec_SOH}, in order to validate the corresponding stability results on the SOH model.

\subsection{Validation of the SOH models.}\label{Sec_val_SOH}

In this section, the SOH model parameters are given as
\begin{equation}\label{model_para}
c_1 = 0.9486, \quad c_2 = 0.8486, \quad \gamma = 0.11857,
\end{equation}
and $d = 0.1$, which are computed using the formulas derived in Appendix \ref{app2} with the parameters of the particle model: $N = 10^5, \nu = 100, D = 10, R_1 = R_2 = 0.1$, so that we can perform the comparison between the two types of models.
Indeed, this choice of parameters would correspond, after the scaling of Appendix \ref{app1}, to a value of $\varepsilon = 0.01$.
The three parameters for the velocity function are chosen as $(\rho^*, \alpha, \beta) = (0.12, 10, 1)$ which results in a stable model. 
This choice is only made for demonstration purposes.

The domain is uniformly partitioned using two integers $N_x$ and $N_y$ which represent the mesh size for the numerical integration in two dimensions.
We apply the splitting method to the SOH model starting with the same initial condition and iteratively calculate the numerical errors with gradually decreasing mesh spacing $(\Delta x = \frac{L_x}{N_x}, \Delta y = \frac{L_y}{N_y})$.
Fig. \ref{accuracy} is the profile of the errors as a function of mesh spacing $\Delta x$ in $\log$ scale and indicates that as expected the method is accurate to the first order.
\begin{figure}[!h]
\begin{center}
\includegraphics[width = 0.50\textwidth]{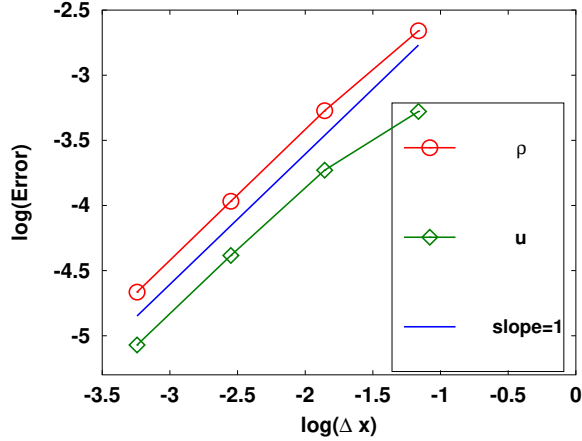}
\end{center}
\caption{The accuracy test at $t=1$ shows that the splitting scheme is of first order. The initial data is given by $\rho_0 = \rho_s(1+0.1\sin(\pi x)), \theta_0 = \theta_s(1+0.1\sin(\pi x))$ with $(\rho_s, \theta_s) = (0.01, \frac{\pi}{4})$ and the mesh sizes are iteratively $N_x = N_y = 32, 64, 128, 256$.}
\label{accuracy}
\end{figure}

We show that the SOH model agrees with the particle model using the following example. 
We construct an initial configuration with $\rho_0(x,y) \equiv 0.01$ and a velocity field given by the Taylor-Green type function
\[
\mathbf{\Omega}_0(x,y) = (\sin\left(\frac{\pi}{5}x\right)\cos\left(\frac{\pi}{5}y\right), 
-\cos\left(\frac{\pi}{5}x\right)\sin\left(\frac{\pi}{5}y\right))^T.
\]
Fig. \ref{part_SOH}(a) shows the contours of the density $\rho$ and the velocity field $\mathbf{\Omega}$ at $t = 0.5$, which are the average of 40 simulations using the particle model with $N = 10^5, \nu = 100, D = 10, R_1 = R_2 = 0.1$. This ensemble average is necessary due to the stochastic nature of both the initial particle positions and the angular dynamics.
Fig. \ref{part_SOH}(b) is the numerical solution produced by the SOH model with \eqref{model_para} and $d=0.1$. 
As explained at the the beginning of Section \ref{Sec_val_SOH}, the parameters for both the Vicsek and SOH models are consistent. The agreement is fairly good. 
\begin{figure}[!h]
\begin{center}
\subfigure[$t = 0.5$: the particle model.]
{\includegraphics[width = 0.49\textwidth]{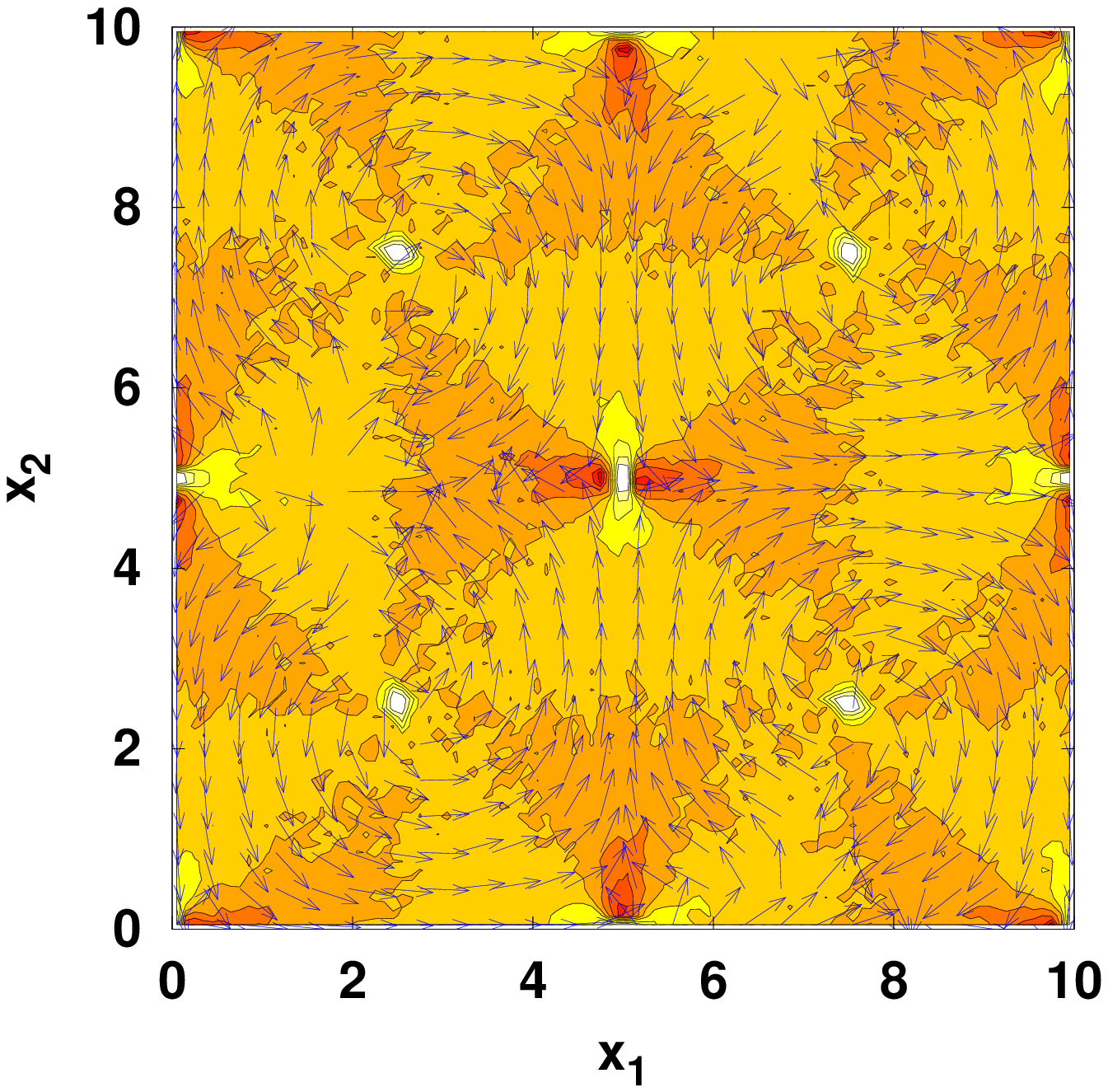}}
\subfigure[$t = 0.5$: the SOH model.]
{\includegraphics[width = 0.49\textwidth]{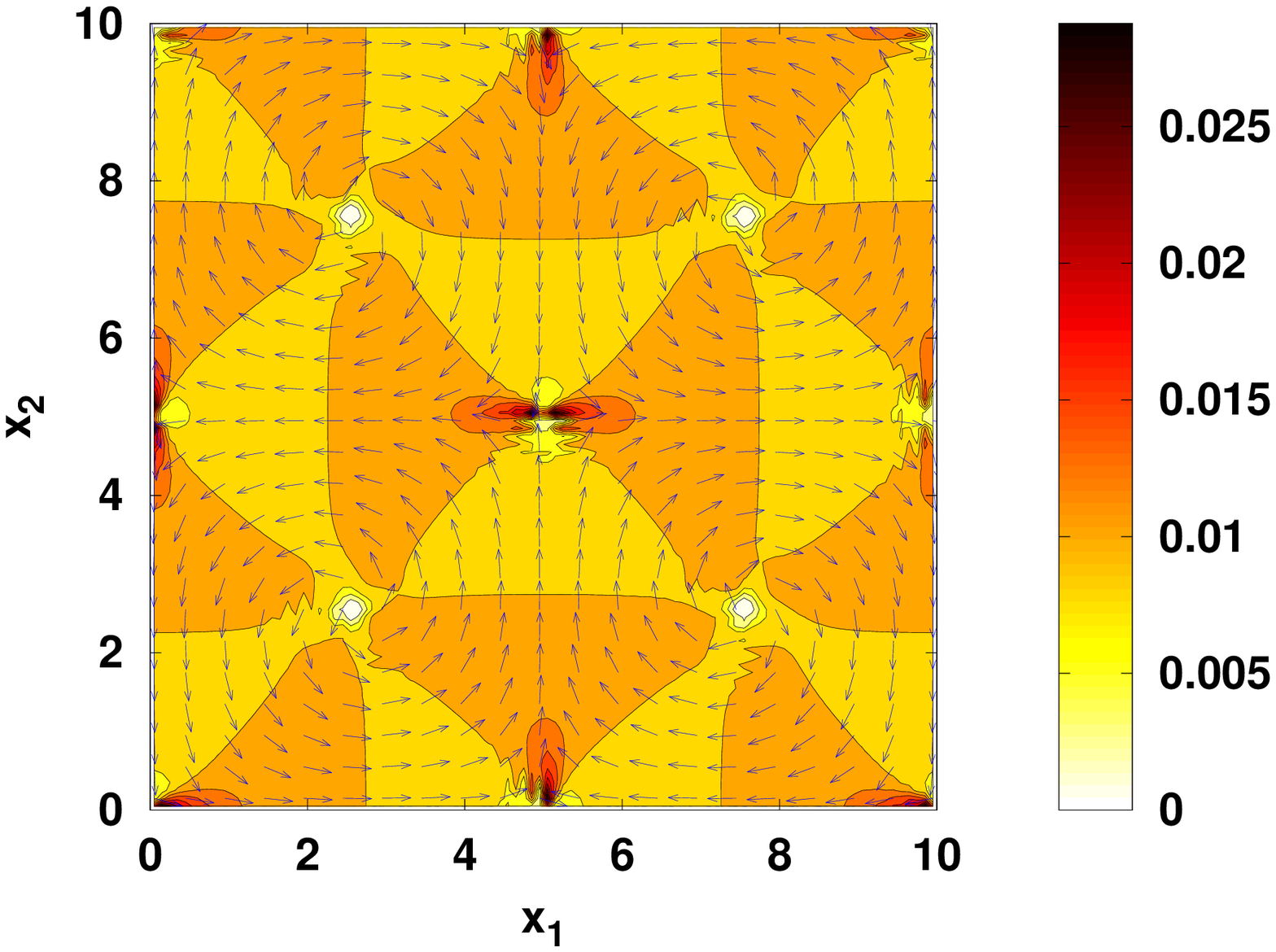}}
\end{center}
\caption{Comparison between the particle (left) and SOH models (right). For the particle model, $N = 10^5, \nu = 100, D = 10, R_1 = R_2 = 0.1$. 
The result is the average of 40 simulations. For the SOH model, $N_x = N_y = 100$.}
\label{part_SOH}
\end{figure}

\subsection{Numerical results for the stability of the models.}

We vary the three parameters in the function $v(\rho)$ to demonstrate the stability of the models around uniform steady states. The model parameters are given by \eqref{model_para}.
We choose an initial condition with a sinusoidal perturbation along $x$ in both density and angle, in phase with each other, $\rho_0 = \rho_s(1+\sigma\sin(\pi x)), \theta_0 = \theta_s(1+\sigma\sin(\pi x))$ with $(\rho_s, \theta_s) = (0.01, \frac{\pi}{4})$.
The model stability will be measured by the Root Mean Square Fluctuation (RMSF) of $(\rho, \theta)$, i.e. the $L^2$ norm of $\rho - \rho_s$ and $\theta-\theta_s$: 
\[
{\rm RMSF}(\rho) = \|\rho-\rho_s\|_{L^2} = \left(\int_0^{L_x}\int_0^{L_y}(\rho(x,t)-\rho_s)^2\,dxdy\right)^\frac12,
\]
and ${\rm RMSF}(\theta)$ is defined in a similar way.

\subsubsection{The inviscid models.}

For this numerical test, let $(\rho^*, \alpha, \beta) = (0.005, 2, 1)$ and $d = 10$.
Again, these choices are only made for demonstration purposes.
Note that $\tilde{v}'(\rho) < 0$ and the inequality on the right-hand side in \eqref{inviscid_Delta_neg} does not hold. 
Therefore the resulting model is unstable and RMSF of $(\rho, \theta)$ must grow in time, which is demonstrated in Fig. \ref{stability_test} with RMSF in $\log$-scale with respect to time $t$. An exponential growth of the perturbation, like we expect, would translate to a straight line in these graphs.
\begin{figure}[!h]
\begin{center}
\subfigure[$\log\big({\rm RMSF}(\rho)\big)$ for $0\leq t \leq 5$.]
{\includegraphics[width = 0.44\textwidth]{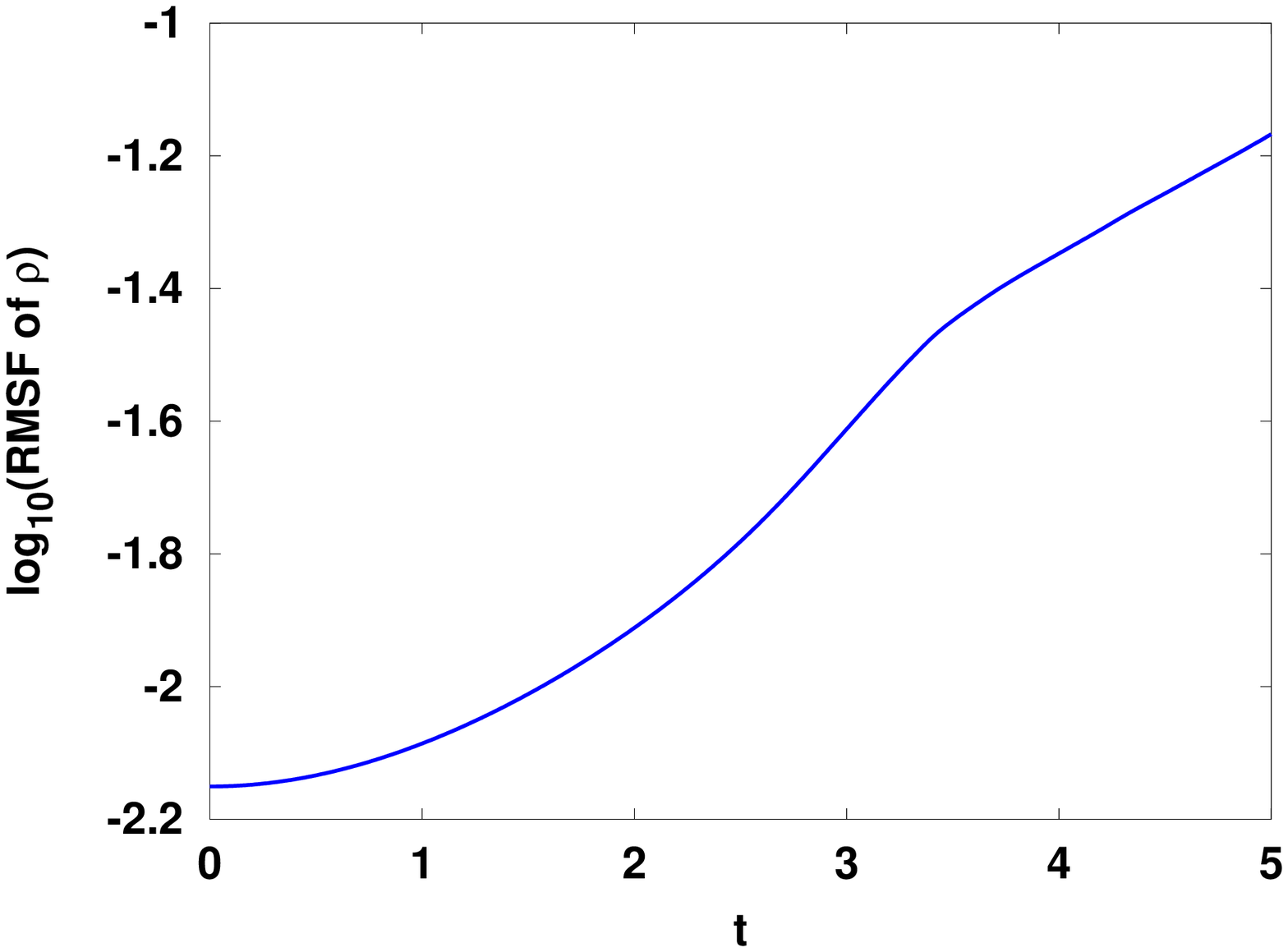}}\qquad
\subfigure[$\log\big({\rm RMSF}(\theta)\big)$ for $0\leq t \leq 5$.]
{\includegraphics[width = 0.44\textwidth]{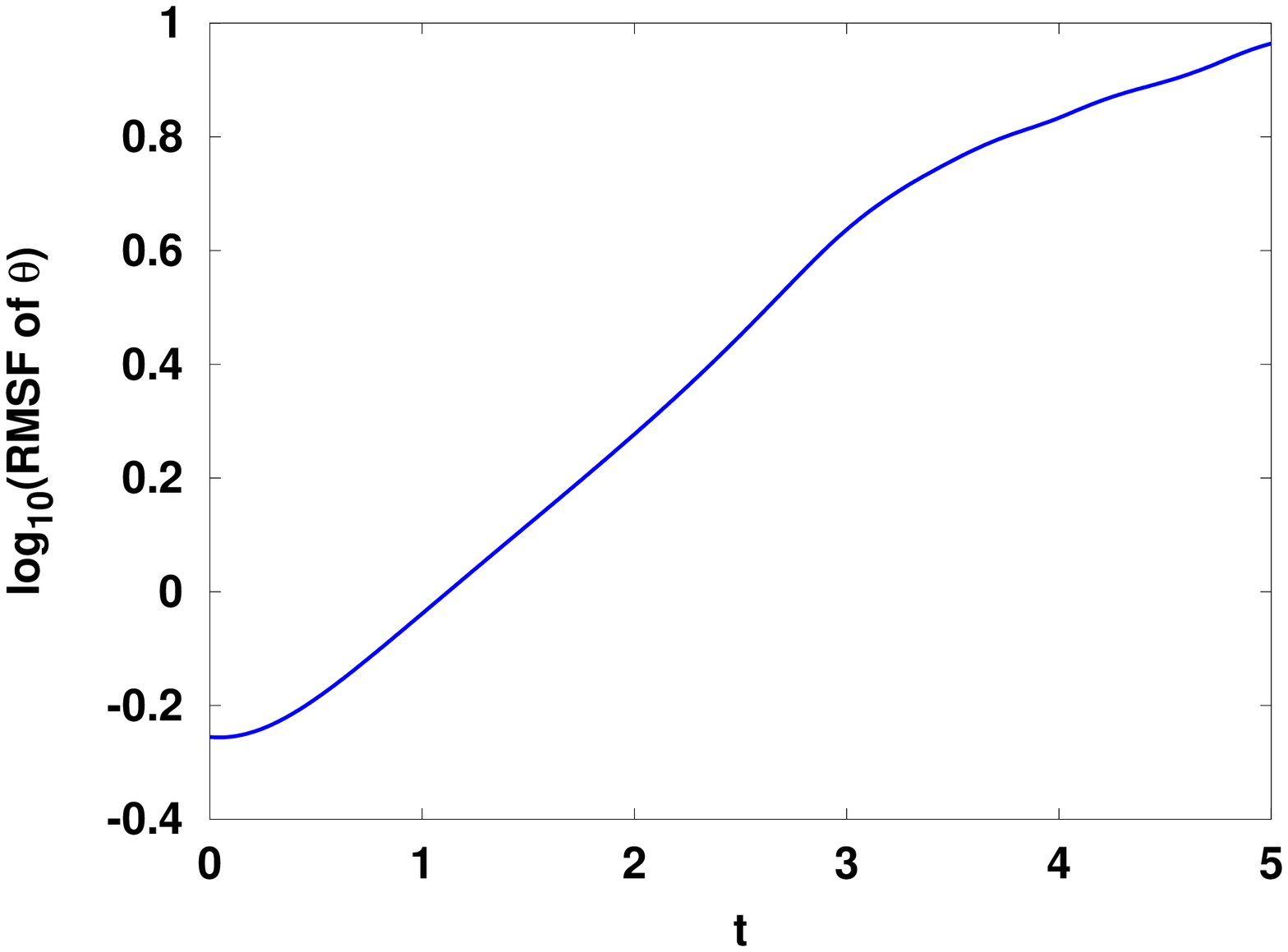}}
\end{center}
\caption{Stability test of the SOH model with $\gamma = 0$ and $\sigma = 0.1$. The profiles show the RMSF of $\rho$ and $\theta$ in $\log$-scale with respect to time $t$.
The numerical solutions evolves from the steady states, and locally high concentrations develop. 
}
\label{stability_test}
\end{figure}

\subsubsection{The viscous models.}
For the viscous models, we will provide a set of comparative examples for the stable and unstable results.
The parameters are taken as \eqref{model_para} and $d = 0.5$. 
For the velocity function, fix $\alpha = 2$ and $\beta = 5$. 
Then the sign of $\tilde{v}'(\rho)$ will depend on the values of $\rho^*$.
From the analysis in Section \ref{Sec_Viscous}, we predict that the model around the constant steady state is stable if $\rho^*$ is large enough.
Fig. \ref{viscous_stable} shows the RMSF of the numerical solutions $(\rho, \theta)$ computed with $\rho^* = 0.02$ and $\sigma = 0.01$.
They are decreasing all the time up to $t = 20$.
\begin{figure}[!h]
\begin{center}
\subfigure[$\log\big({\rm RMSF}(\rho)\big)$ for $0 \leq t \leq 20$.]
{\includegraphics[width = 0.44\textwidth]{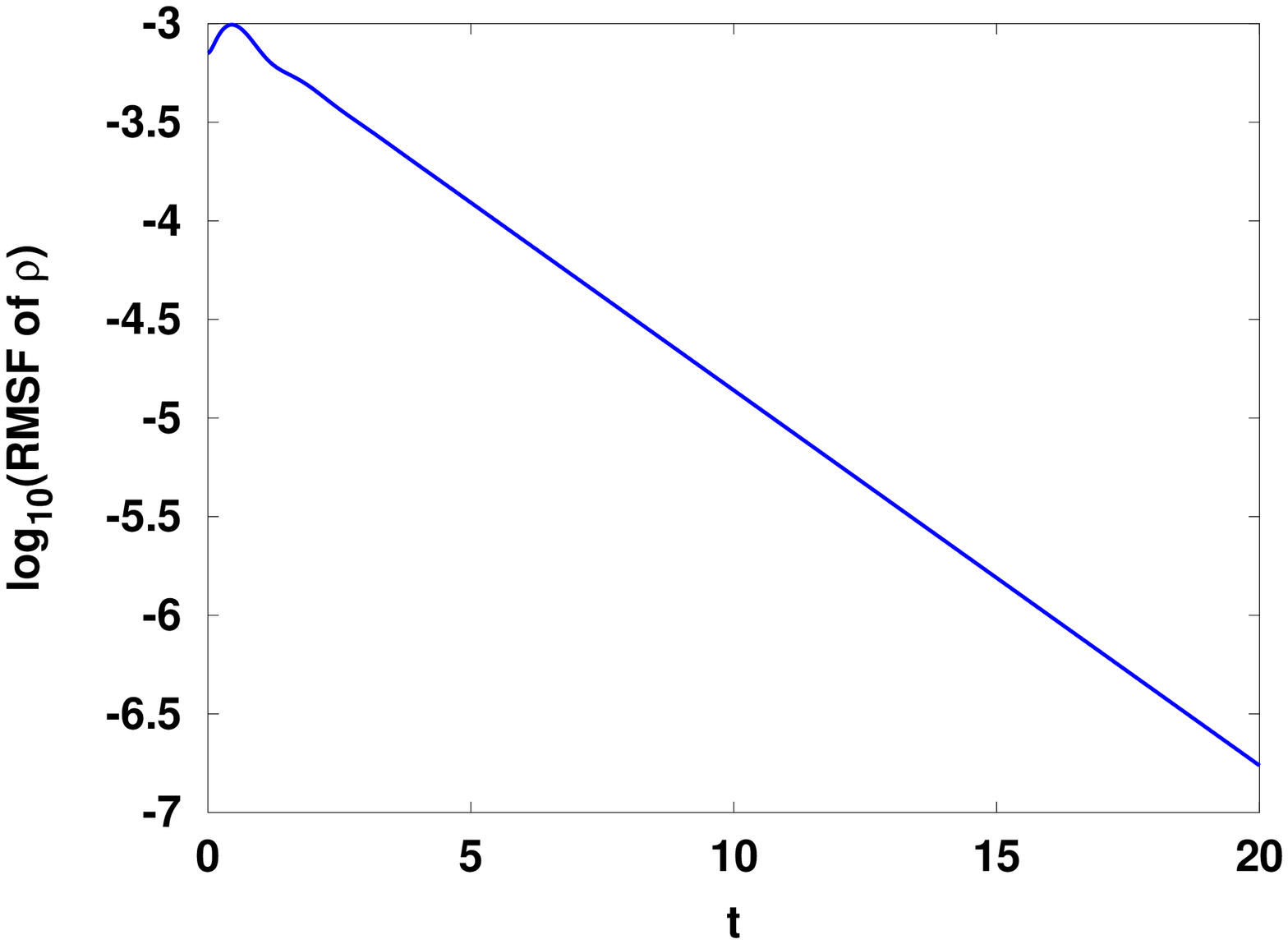}}\qquad
\subfigure[$\log\big({\rm RMSF}(\theta)\big)$ for $0\leq t \leq 20$.]
{\includegraphics[width = 0.44\textwidth]{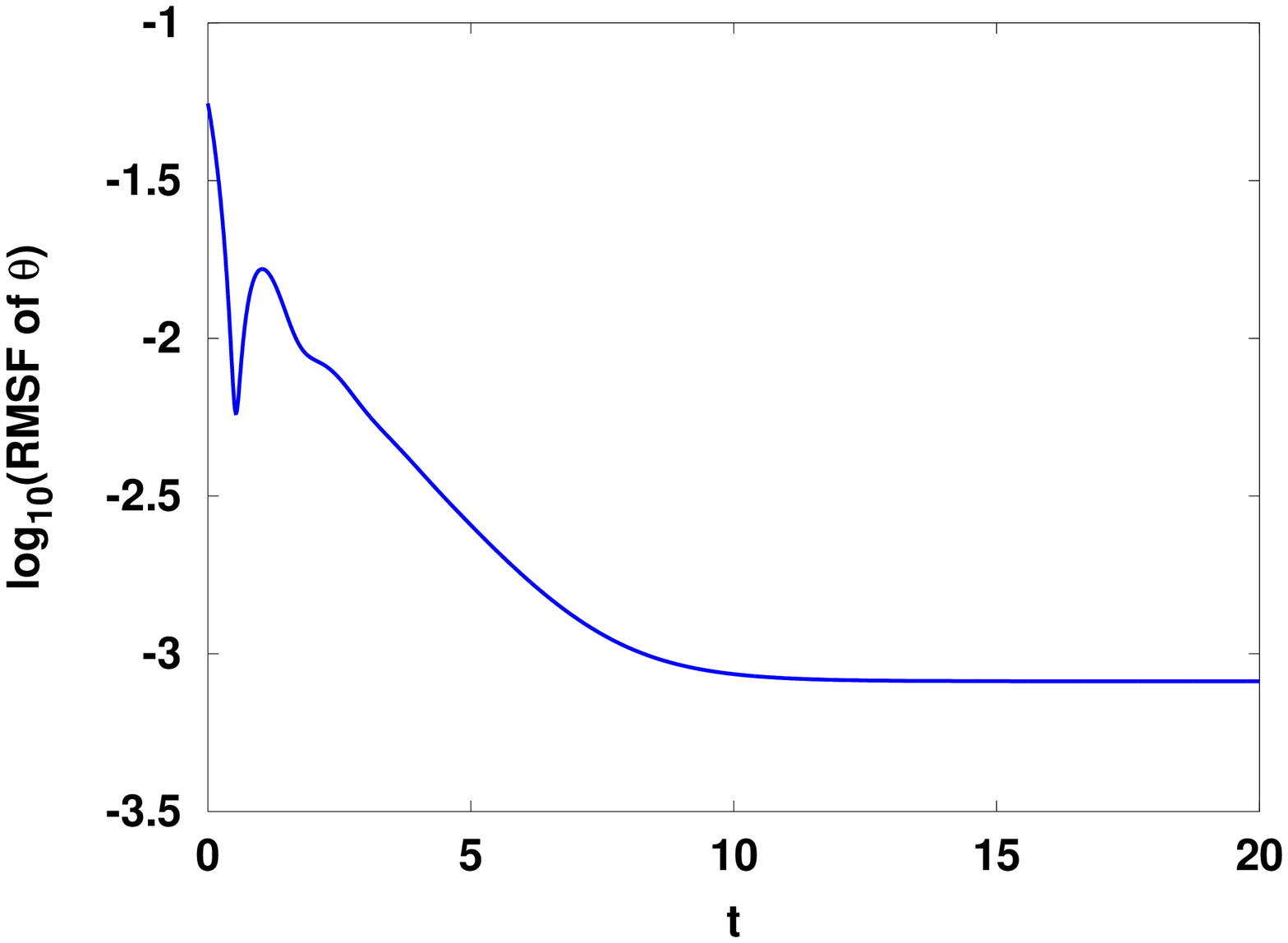}}
\end{center}
\caption{Stability test of the SOH model with $d = 0.5$, $\sigma = 0.01$ and $\rho^* = 0.02$.
RMSF of the density and angle are decreasing with time. }
\label{viscous_stable}
\end{figure}

Let $\rho^* = 0.005$ where we have $\tilde v'(\rho) < 0$ at the initial time. 
Hence the numerical solutions are expected to evolve away from $(\rho_s, \theta_s)$, indicating the instablility of the model.
Fig. \ref{viscous_instable} shows the RMSF of the numerical solutions $(\rho, \theta)$ for the time interval $t\in [0,5]$. 
One can observe that they grow significantly; compare with the numerical solutions in Fig. \ref{viscous_stable} whose RMSF decreases to zero.
\begin{figure}[!h]
\begin{center}
\subfigure[$\log\big({\rm RMSF}(\rho)\big)$ for $0\leq t \leq 15$.]
{\includegraphics[width = 0.44\textwidth]{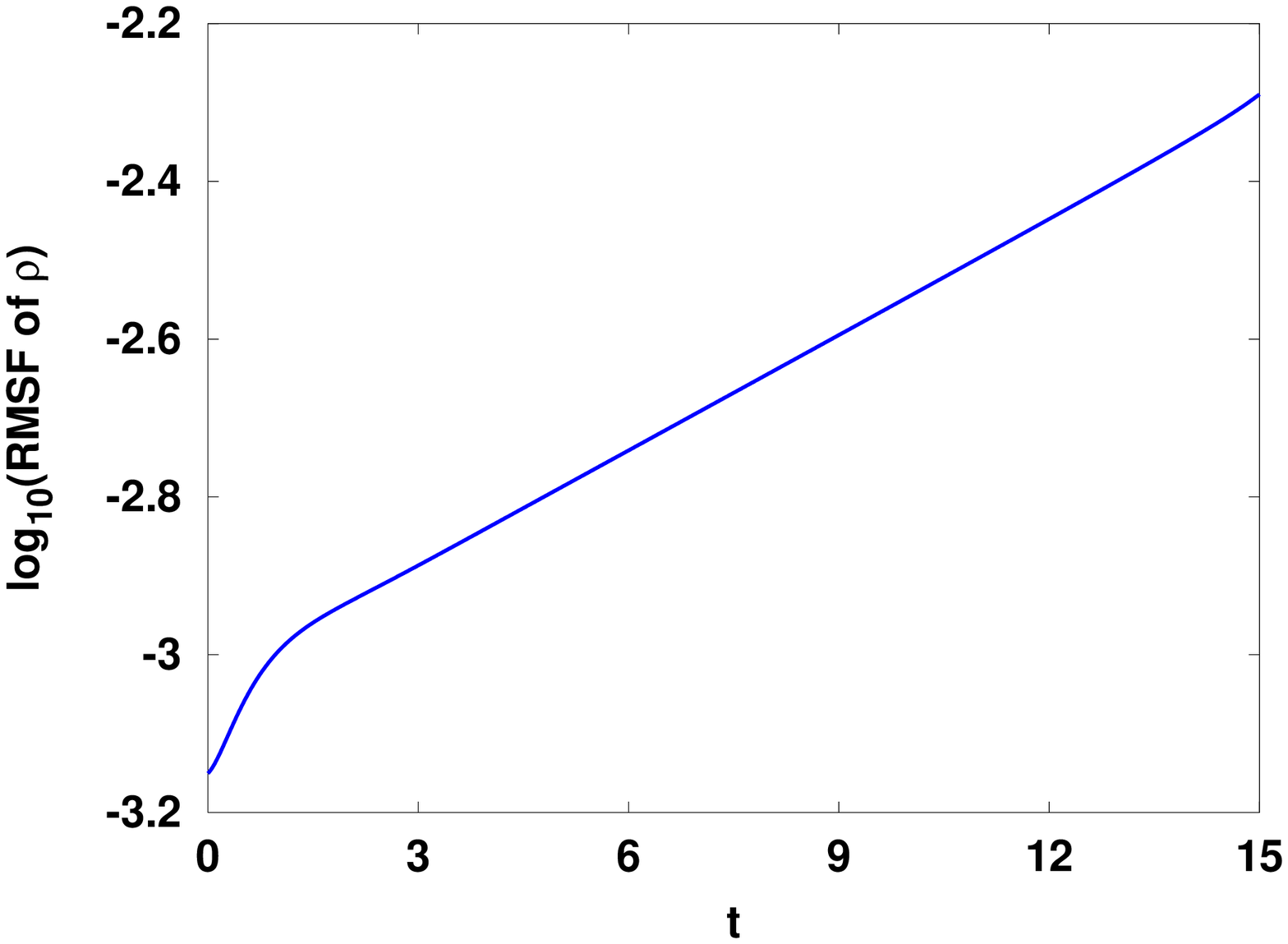}}\qquad
\subfigure[$\log\big({\rm RMSF}(\theta)\big)$ for $0\leq t \leq 15$.]
{\includegraphics[width = 0.44\textwidth]{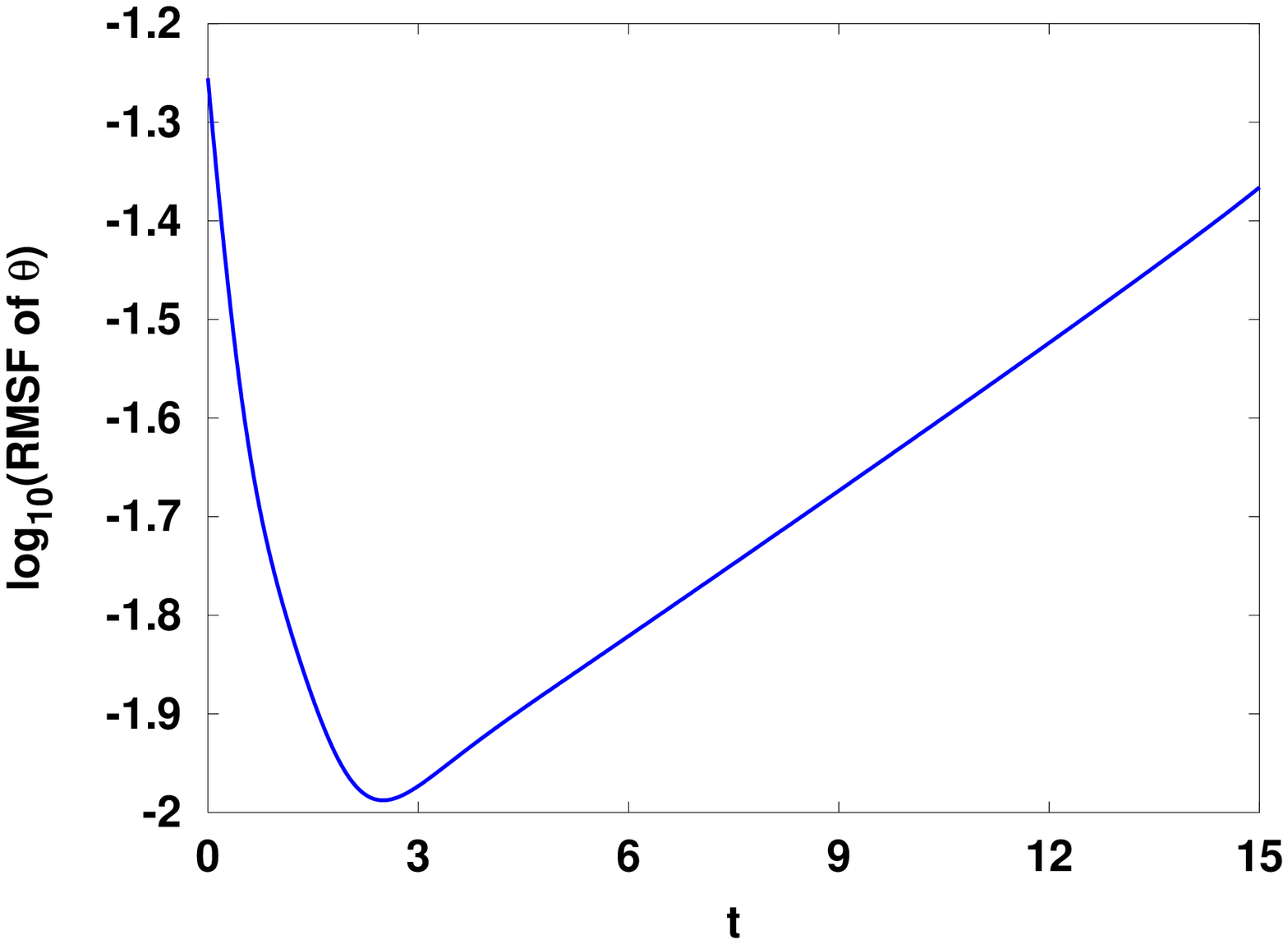}}
\end{center}
\caption{Stability test of the SOH model with $d = 0.5$, $\sigma = 0.01$ and $\rho^* = 0.005$.
The RMSF of the numerical solutions grows gradually and high local concentrations develop. 
The linear scaling of the $\log$ of RMSF implies an exponential growth of the perturbation as a function of time $t$.} 
\label{viscous_instable}
\end{figure}

\subsection{Growth rate of the instability.}
 
The modal analysis of the hydrodynamic system \eqref{sys_FT} shows that the magnitude of $\hat\rho_\sigma$ possesses an exponential growth rate $\xi{\rm Im}\lambda$.
In this section, we will examine the numerical solutions provided by the particle and SOH models and compare the numerical growth rate of the perturbation to the modal analysis result of the viscous system \eqref{sys_FT}. 
For the sake of convenience, we perform a Discrete Fourier Transform (details can be found in Appendix \ref{viscous_DFT}).
The initial configuration is taken as:
\begin{align*}
\rho_0(x,y) = \rho_s(1+\sigma\rho_\sigma(x,y)), \qquad
\theta_0(x,y) = \theta_s(1+\sigma\theta_\sigma(x,y)),
\end{align*}
where $\sigma = 0.01$, and $\rho_s = 0.01$ is fixed. $\rho_\sigma$ takes the form
\begin{align*} 
\rho_{\sigma}(x) = \sum_{\xi=0}^{10}\left[a_1(\xi) \cos\left(\frac{2\pi\xi}{L_x}\big(x - \frac{\Delta x}{2}\big)\right) 
+ a_2(\xi)\sin\left(\frac{2\pi\xi}{L_x}\big(x - \frac{\Delta x}{2}\big)\right)\right],
\end{align*}
and $\theta_\sigma$ is given in a similar way.
Here $a_1(\xi), a_2(\xi)$ in $\rho_\sigma$ and $\theta_\sigma$ are different random numbers generated from a uniform distribution on the interval $[0,1]$.

The parameters of the SOH model are fixed as $c_1 = 0.975, c_2 = 0.925, d = 0.05$ and $\gamma = 0.12188$. 
They correspond to the particle parameters where $N = 10^5, \nu = 100, D = 5$, and $R_1 = R_2 = 0.1$.
The three parameters for $v(\rho)$ are chosen as $(\rho^*, \alpha, \beta) = (0.005, 2, 5)$ in which case the viscous model is unstable.
The steady state for the density is chosen as $\rho_s = 0.01$. 
We vary the steady state orientation $\theta_s$ in the interval $[0, \frac{\pi}{2}]$ and plot the growth rate $\xi{\rm Im}\lambda$ as a function $\xi \in [0, 6]$ for both the linearized and SOH models in Fig. \ref{viscous_contour}.
Fig. \ref{viscous_contour}(a) is computed using the fomula \eqref{eq_lambda}. 
To obtain Fig. \ref{viscous_contour}(b), we proceed for each $\theta_s$ in the following way. We compute the numerical solutions of the SOH model up to time $t=1$, and perform the Discrete Fourier transform on the perturbed part $\rho_\sigma = \rho - \rho_s$ to get $\hat\rho_\sigma(\xi,t)$. 
With different initial data, i.e. different $a_1$ and $a_2$, we collect $N_{\text{sam}}$ samples of $\hat\rho_\sigma(\xi,t)$ and apply a simple linear regression to the averaged quantities
\[
\frac{1}{N_{\text{sam}}}\sum_{N_{sam}}\frac{\hat\rho_\sigma(\xi,t)}{\hat\rho_\sigma(\xi,0)}
\]
with respect to time $t$. This will generate the growth rate for each $\xi$; Fig. \ref{viscous_contour}(b) shows the results for $\xi = 0, 1, \ldots, 6$ with $N_{\text{sam}} = 100$. 
The motivation here is as follows: by choosing random coefficients for the different modes we generate a set of initial conditions that have statistically equal weight for each mode. 
The growth rate is interpreted as the slope of the function $t \to \hat\rho_\sigma(\xi, t)$ in log scales. 
The similarity of these two sets of contours is the increase of the growth rate with respect to both $\xi$ and $\theta_s$. The difference is that the fully nonlinear SOH model rapidly develops much larger growth rates compared to the linearised solution. There are also fluctuations at certain $\theta_s$ and $\xi$ likely due to finite size effects. 
\begin{figure}[!h]
\begin{center}
\subfigure[The linear prediction.]
{\includegraphics[width = 0.48\textwidth]{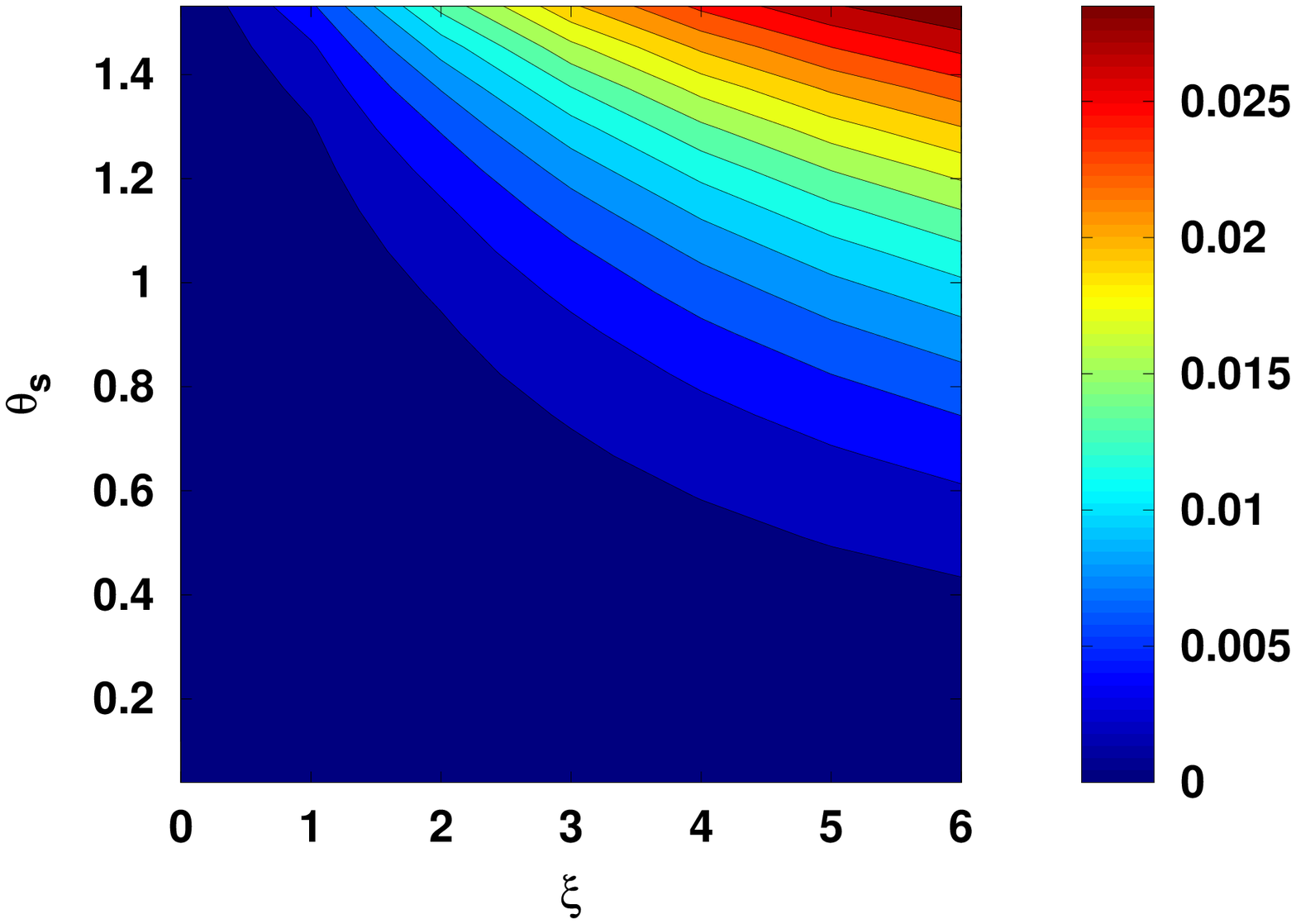}}
\subfigure[The SOH model.]
{\includegraphics[width = 0.48\textwidth]{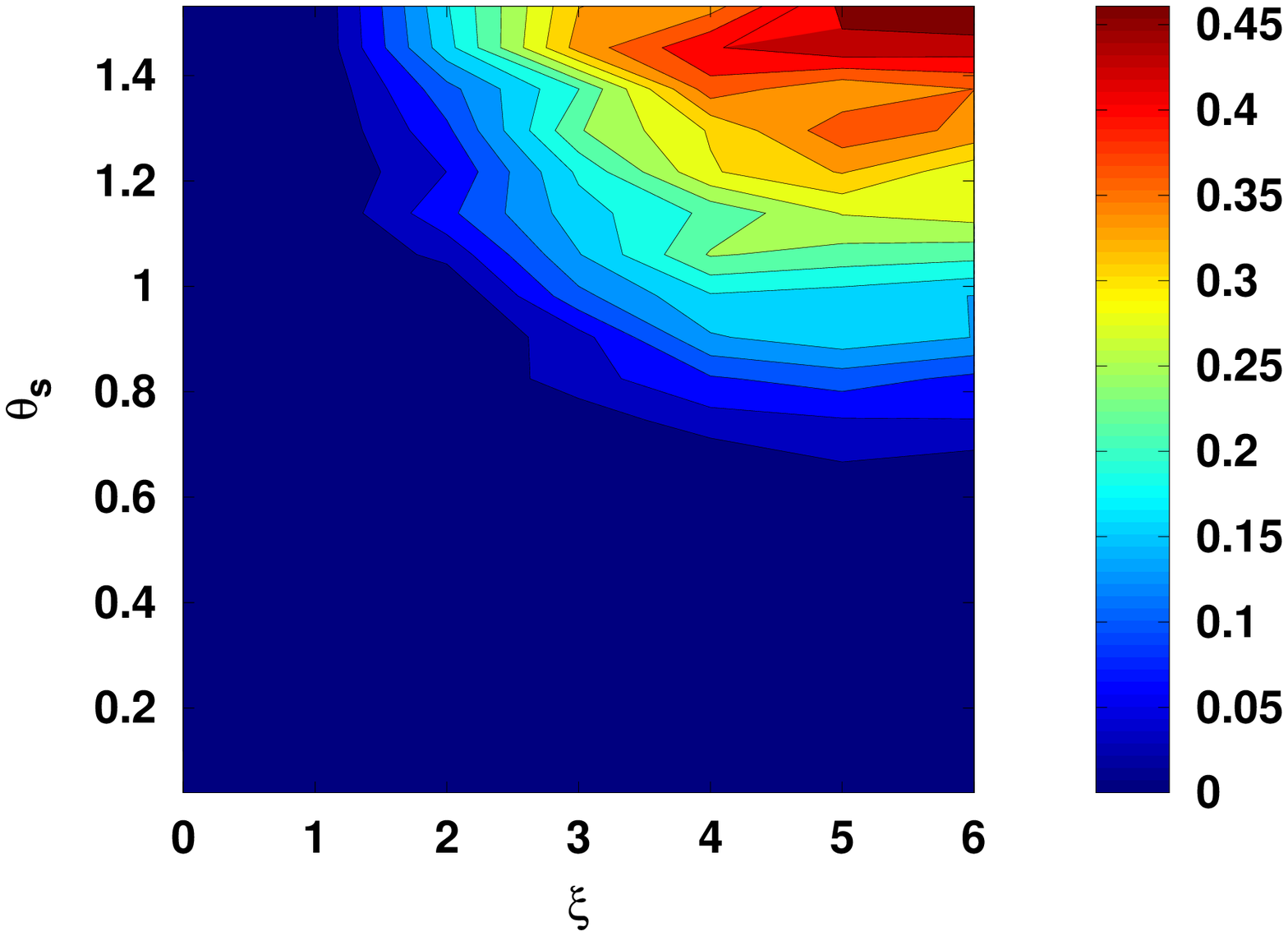}}
\end{center}
\caption{Growth rate of the perturbation $\rho_\sigma$. 
The parameters are $c_1 = 0.975, c_2 = 0.925, d = 0.05$ and $k_0 = 0.125$.
The three parameters for $v(\rho)$ are chosen as $(\rho^*, \alpha, \beta) = (0.005, 2, 5)$.
The steady state for the density is fixed at $\rho_s = 0.01$ and the final time is $t = 1$.
(a) is computed using the fomula \eqref{eq_lambda}. 
In order to obtain (b), we compute the numerical solutions of the SOH model and perform a simple linear regression on the Discrete Fourier transform of the perturbed part, i.e. $\rho_\sigma = \rho - \rho_s$. 
The growth rate is interpreted as the slope of the function $t \to \hat\rho_\sigma(\xi, t)$. }
\label{viscous_contour}
\end{figure}

\begin{figure}[!h]
\begin{minipage}[c]{0.45\textwidth}
\includegraphics[width = 1.00\textwidth]{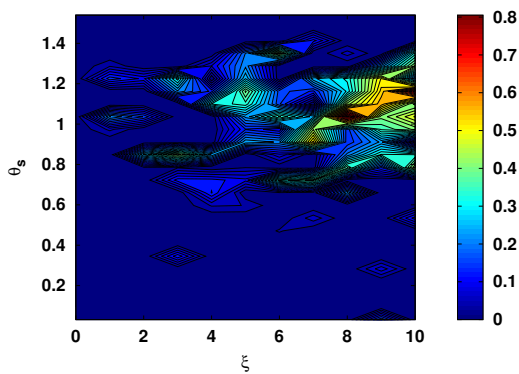}
\end{minipage}
\begin{minipage}[c]{0.55\textwidth}
\caption{Growth rate of the perturbation $\rho_\sigma$ given by the particle model with
$N = 10^5$. The parameters are $\nu = 100, D = 5, R_1 = R_2 = 0.1$.
The three parameters for $v(\rho)$ are chosen as $(\rho^*, \alpha, \beta) = (0.005, 2, 5)$.
For each $\theta_s$, the growth rate is computed using the average of $10$ simulations, in order to reduce the effects of noise.
}
\label{ptc_contour}
\end{minipage}
\end{figure}

Fig. \ref{ptc_contour} shows the growth rate of the perturbation $\rho_\sigma$ given by the particle model.
The number of particles is $N = 10^5$. The parameters are given as $\nu = 100, D = 5, R_1 = R_2 = 0.1$. 
And they match the parameters for the SOH model in Fig. \ref{viscous_contour}.
The three parameters for $v(\rho)$ are chosen as $(\rho^*, \alpha, \beta) = (0.005, 2, 5)$ to match the SOH model. 
For each $\theta_s$, the growth rate is computed using the average of $10$ simulations, in order to reduce the effects of noise.
Although the contours do not exhibit the monotonic behaviour of the growth rate of the SOH model with respect to the steady state angle $\theta_s$ and the eigenmode $\xi$ due to the nonlinearity and stochastic effects, one can observe the stronger instability for larger $\theta_s$ and $\xi$.
 
\section{Conclusion.}
We have studied a Vicsek model where the velocity depends on the local density and then derived the corresponding SOH model.
At the hydrodynamic scale, we analyse the stability of the two-dimensional inviscid and viscous models around their steady states.
In summary, the stability of the SOH models is determined by the behaviour of the mass flux $\rho v(\rho)$. 
The theoretical results are illustrated by the numerical simulations with different choices of the velocity function. In general, we find good agreement between the theoretical prediction of the onset of instability and the numerical results. In the unstable regime, while our numerical results are qualitatively compatible with the predictions, the strong non-linearities present in the SOH model and especially the particle model quickly dominate the response.
The SOH model we have developed here is a useful model to describe the semen flow in the experiments designed in \cite{Creppy_etal_2015}. These experiments record the correlation between the averaged velocity and the density of the sperm cells,
and the velocity as a function of the density was fitted using the SOH model.
Further study of vortices observed in the collective behaviour of the semen flow is under way.
\begin{appendices}
\setcounter{equation}{0}
\section{Derivation of the SOH model with nonconstant velocity.}\label{SOH_derivation}
The SOH model can be explicitly coarse-grained from the particle model introduced in Eq. \eqref{IBM}.
For completeness, we provide the main steps here and readers are referred to \cite{Degond_Motsch_M3AS08} for details.
\subsection{The mean field model.}\label{app1}
We consider the limit of the system when $N \to \infty$. Introduce the empirical distribution $f^N(x,\omega,t)$ defined as
\begin{equation}
f^N(x,\omega,t) = \frac{1}{N}\sum_{i=1}^N\delta(x-X_i(t))\delta(\omega,\omega_i(t)),
\end{equation}
where $\omega \in \mathbb{S}^{n-1}$, the unit sphere in $\mathbb{R}^n$ and the distribution $\delta(\omega,\omega')$ is defined as
\[
\langle \delta(\omega,\omega'),\varphi(\omega)\rangle = \varphi(\omega') \quad \text{ for any smooth function } \varphi.
\]

Sending $N \to \infty$ and scaling out to obtain dimensionless parameters, 
the formal mean-field system for the probability distribution function $f(x,\omega,t)$ on $\mathbb{R}^n\times\mathbb{S}^{n-1}\times(0,\infty)$ 
 is given by
\begin{equation}
\partial_t f + \nabla_x\cdot(v(m_f)\omega f) + \nabla_{\omega}\cdot(Gf) = \breve D\Delta_{\omega} f,
\end{equation}
where $\Delta_{\omega}$ denotes the Laplace-Beltrami operator on the sphere and
\begin{align*}
G(x,\omega,t) &= \breve \nu\mathcal{P}_{\omega^\perp}\bar{\omega}(x,\omega,t) 
= \breve \nu\mathcal{P}_{\omega^\perp}\frac{\mathcal{J}_f(x,t)}{|\mathcal{J}_f(x,t)|},\\
\mathcal{J}_f(x,t) &= \int_{\mathbb{R}^n\times\mathbb{S}^{n-1}}K_1\left(\frac{|x-y|}{\breve R_1}\right) \omega f(y,\omega,t)\,dyd\omega,\\
m_f(x,t) &= \frac{1}{|B_{\breve R_2}|}\int_{\mathbb{R}^n \times\mathbb{S}^{n-1}} K_2\left(\frac{|x-y|}{\breve R_2}\right) f(y,\omega,t)\,dyd\omega.
\end{align*}
Here we assume that $K_1$ and $K_2$ only depend on the distance between particles, characterized by the dimensionless parameters $\breve R_1$ and $\breve R_2$. 
We also assume that both of the ranges of the interaction kernels $K_1$ and $K_2$ are small. 

Let $\varepsilon$ be a small positive number. We will perform the explicit hydrodynamic limit by introducing rescaled parameters that individually tend to zero in the limit $\epsilon\rightarrow 0$.
Then $\breve R_1 = \sqrt{\varepsilon}\hat R_1, \breve R_2 = \sqrt{\varepsilon}\hat R_2$. 
We also assume that the alignment strength (a.k.a the social forces), and the diffusion coefficient are large, but of similar magnitude i.e.,
\[
\breve\nu = \frac{1}{\varepsilon}, \quad \frac{\breve D}{\breve\nu} = d = \mathcal{O}(1).
\]
For simplicity, we drop the hats and have the following result:
\begin{lemma}
The density $f^\varepsilon(x,\omega,t)$ satisfies the following equation:
\begin{equation}
\varepsilon[\partial_t f^{\varepsilon} + \nabla_x\cdot(v(\rho_{f^\varepsilon})\omega f^\varepsilon) ]
+\nabla_\omega\cdot[\mathcal{P}_{\omega^\perp}(\mathbf\Omega_{f^\varepsilon} + \varepsilon\mathbf\Omega_{f^\varepsilon}^1)f^\varepsilon] 
= d\Delta_\omega f^\varepsilon + \mathcal{O}(\varepsilon^2),
\end{equation}
where
\begin{align*}
\mathbf\Omega_{f^{\varepsilon}} &=\frac{J_{f^\varepsilon}}{|J_{f^\varepsilon}|} \text{ with } 
J_{f^{\varepsilon}} = \int_{\mathbb{S}^{n-1}}vf^{\varepsilon}(x,\omega,t)\,d\omega,\\
\mathbf\Omega_{f^\varepsilon}^1 &=\frac{k_1^0}{|J_{f^\varepsilon}|}\mathcal{P}_{\mathbf\Omega_{f^\varepsilon}^\perp}\Delta_xJ_{f^\varepsilon}
\text{ with } k_1^0 = \frac{R_1^2}{2n}\frac{\int_{\mathbb{R}^n}K_1(|z|)|z|^2\,dz}{\int_{\mathbb{R}^n}K_1(|z|)\,dz},\\
\rho_{f^\varepsilon}(x,t) &= \int_{\mathbb{S}^{n-1}}f^\varepsilon(x,\omega,t)\,d\omega.
\end{align*}
\end{lemma}

We drop the higher order term of $\varepsilon$, $\mathcal{O}(\varepsilon^2)$ and define the collisional operator $\mathcal{Q}(f^\varepsilon)$ by
\begin{equation}
\mathcal{Q}(f^\varepsilon) = -\nabla_\omega\cdot\mathcal{P}_{\omega^\perp}\mathbf\Omega_{f^\varepsilon}f^\varepsilon + d\Delta_\omega f^\varepsilon.
\end{equation}
The rescaled system can be written as
\begin{equation}\label{scaled_1}
\varepsilon[\partial_t f^{\varepsilon} + \nabla_x\cdot(v(\rho_{f^\varepsilon})\omega f^\varepsilon) 
+\nabla_\omega\cdot\mathcal{P}_{\omega^\perp}\mathbf\Omega_{f^\varepsilon}^1f^\varepsilon] 
= \mathcal{Q}(f^\varepsilon),
\end{equation}
where
\begin{align}
\mathbf\Omega_{f^{\varepsilon}} &=\frac{J_{f^\varepsilon}}{|J_{f^\varepsilon}|} \text{ with } 
J_{f^{\varepsilon}} = \int_{\mathbb{S}^{n-1}}\omega f^{\varepsilon}(x,\omega,t)\,d\omega,\\
\mathbf\Omega_{f^\varepsilon}^1 &=\frac{k_1^0}{|J_{f^\varepsilon}|}\mathcal{P}_{\mathbf\Omega_{f^\varepsilon}^\perp}\Delta_xJ_{f^\varepsilon}
\text{ with } k_1^0 = \frac{R_1^2}{2n}\frac{\int_{\mathbb{R}^n}K_1(|z|)|z|^2\,dz}{\int_{\mathbb{R}^n}K_1(|z|)\,dz},\\
\rho_{f^\varepsilon}(x,t) &= \int_{\mathbb{S}^{n-1}}f^\varepsilon(x,\omega,t)\,d\omega. \label{scaled_4}
\end{align}
\subsection{The hydrodynamics model.} \label{app2}
This section derives the hydrodynamic model for the local density $\rho_f$ and the local mean orientation $\mathbf\Omega_f$ 
which will be valid at the macroscopic scale 
by taking the limit of the system (\ref{scaled_1})-(\ref{scaled_4}) as $\varepsilon \to 0$.

We first introduce the von Mises-Fisher (VMF) probability distribution $M_\mathbf\Omega(\omega)$ on $\mathbb{S}^{n-1}$ associated to a given $\mathbf\Omega \in\mathbb{S}^{n-1}$:
\begin{equation}
M_{\mathbf\Omega} = \frac{1}{Z} \exp\left(\frac{\omega\cdot\mathbf\Omega}{d}\right) 
\text{ with } Z = \int_{\mathbb{S}^{n-1}}\exp\left(\frac{\omega\cdot\mathbf\Omega}{d}\right)\,d\omega.
\end{equation}
The main result in this section is the following theorem:
\begin{theorem}
Let $f^\varepsilon$ be the solution of (\ref{scaled_1})-(\ref{scaled_4}). Assume that there exists $f$ such that
\begin{equation}
\lim_{\varepsilon\to 0} f^\varepsilon = f
\end{equation}
pointwise and the limit holds for its derivatives. Then there exist $\rho(x,t)$ and $\mathbf\Omega(x,t)$ such that
\begin{equation}
f(x,\omega,t) = \rho(x,t)M_{\mathbf\Omega(x,t)}(\omega)
\end{equation}
and $(\rho, \mathbf\Omega)$ are the solutions of
\begin{subequations}\label{SOH_model}
\begin{numcases}{}
\partial_t \rho + \nabla_x\cdot(c_1v(\rho)\rho\mathbf\Omega) = 0,\\
\rho[\partial_t\mathbf\Omega + (c_2v(\rho)\mathbf\Omega\cdot\nabla_x)\mathbf\Omega] + d\mathcal{P}_{\mathbf\Omega^\perp}\nabla_x(v(\rho)\rho)
= \gamma\mathcal{P}_{\mathbf\Omega^\perp}\Delta_x(\rho\mathbf\Omega), \\
|\mathbf\Omega| = 1,
\end{numcases}
\end{subequations}
where
\begin{align*}
c_1(d) &= \int_{\mathbb{S}^{n-1}}M_\mathbf\Omega(\omega)(\omega\cdot\mathbf\Omega)\,d\omega,\\
c_2(d) &= \frac{\langle\sin^2\theta\cos\theta h\rangle_{M_\mathbf\Omega}}{\langle\sin^2\theta h\rangle_{M_\mathbf\Omega}},\\
\gamma &= k_1^0[(n-1)d + c_2]. 
\end{align*}
Here $\langle\cdot\rangle_{M_\mathbf\Omega}$ denotes the integration with the weight function $M_{\mathbf\Omega}$ with respect to $\theta$ on the domain $[0,\pi]$.
\end{theorem}
\begin{proof}
The proof is divided into three steps: (i) the determination of the equilibria; (ii) the Generalized Collision Invariants; 
(iii) the hydrodynamic limit. The three subsections below give a sketch of the proof.

{\bf Step 1.} The equilibrium states, i.e., the null space of $\mathcal{Q}$.
\begin{definition}
The set $\mathcal{E}$ of the equilibrium of $\mathcal{Q}$ is given by
\begin{equation}
\mathcal{E} = \{f\in H^1(\mathbb{S}^{n-1})| f\geq 0 \text{ and } \mathcal{Q}(f) = 0\}.
\end{equation}
\end{definition}
We will prove that the set $\mathcal{E}$ consists of the VMF distribution. 
More precisely, we have the following result:
\begin{lemma}\label{lemma_E}
Under certain regularity assumptions, the set of the equilibria $\mathcal{E}$ is given by
\begin{equation}
\mathcal{E} = \{ f(\omega) = \rho M_{\mathbf\Omega}(\omega) \text{ for arbitrary } \rho \geq 0\}.
\end{equation}
\end{lemma}
\begin{proof}
Please refer to \cite{Degond_Motsch_M3AS08} for the proof of Lemma \ref{lemma_E}. 
\end{proof}

{\bf Step 2.} The generalized collision invariants (GCI).
\begin{definition}
A collision invariant (CI) is a function $\psi(\omega)$ such that for any function $f(\omega) \geq 0$ with sufficient regularity we have
\begin{equation}
\int_{\mathbb{S}^{n-1}} \mathcal{Q}(f)\psi\,d\omega = 0.
\end{equation}
We denote by $\mathcal{C}$ the set of CI, the set $\mathcal{C}$ is a vector space. 
\end{definition}
Due to the lack of physical conservation laws except for the total mass, the set of CI is not large enough to allow us to derive the evolution of the macroscopic quantities $\rho$ and $\mathbf\Omega$.
To overcome this difficulty, a weaker concept of collision invariant, the so-called ``Generalized Collision Invariant" (GCI) has been introduced in \cite{Degond_Motsch_M3AS08}. 
We define the collision operator $\mathcal{Q}(\mathbf\Omega, f)$ such that for a given vector $\mathbf\Omega \in \mathbb{S}^{n-1}$, we have
\begin{equation}
\mathcal{Q}(\mathbf\Omega, f) = \nabla_{\omega}\cdot\left[M_{\mathbf\Omega}\nabla_\omega\left(\frac{f}{M_\mathbf\Omega}\right)\right].
\end{equation}
Notice that
\begin{equation}
\mathcal{Q}(f) = \mathcal{Q}(\mathbf\Omega_f,f).
\end{equation}
Then we have
\begin{definition}
Given $\mathbf\Omega \in \mathbb{S}^{n-1}$, a Generalized Collision Invariant (GCI) associated to $\mathbf\Omega$ is a function $\psi \in H^1(\mathbb{S}^{n-1})$ satisfying
\begin{equation}
\int_{\mathbb{S}^{n-1}}\mathcal{Q}(\mathbf\Omega,f)\psi(\omega)\,d\omega = 0\quad
\forall f\in H^1(\mathbb{S}^{n-1}) \text{ with } \mathbf\Omega_f = \pm\mathbf\Omega.
\end{equation}
The set of GCIs associated to $\mathbf\Omega$ is denoted by $\mathcal{C}_\mathbf\Omega$.
\end{definition}
The following lemma characterizes the set of generalized collision invariants.
\begin{lemma}\label{lemma_GCI}
The set $\mathcal{C}_\mathbf\Omega$ is given by
\begin{equation}
\mathcal{C}_\mathbf\Omega = \{ h(\omega\cdot\mathbf\Omega)\beta\cdot\omega + C\text{ where } \beta\in\mathbb{R}^n \text{ with } \beta\cdot\mathbf\Omega = 0
\text{ and } C\in\mathbb{R} \text{ is arbitrary.}\},
\end{equation}
and the scalar function $h(\cdot)$ is such that $h(\cos\theta) = \frac{g(\theta)}{\sin\theta}$ and $g(\theta)$ is the unique solution in the space 
\[
V = \{ g| (n-2)\sin^{\frac{n}{2}-2}\theta g \in L^2(0,\pi), \sin^{\frac{n}{2}-1}\theta g \in H_0^1(0,\pi)\}
\]
of the problem
\begin{equation}\label{ODE_g}
-\sin^{2-n}\theta e^{-\frac{\cos\theta}{d}}\frac{d}{d\theta}\left(\sin^{n-2}\theta e^{\frac{\cos\theta}{d}}\frac{dg}{d\theta}\right)
+\frac{n-2}{\sin^2\theta}g = \sin\theta.
\end{equation}
The set $\mathcal{C}_\mathbf\Omega$ is an $n$-dimensional vector space.
\end{lemma}
\begin{proof}
Please refer to \cite{Degond_Motsch_M3AS08,Frouvelle_M3AS12} for the proof of the above lemma.
\end{proof}

{\bf Step 3.} The hydrodynamic limit.

Integrating Eq. \eqref{scaled_1} against the collision invariants and taking the formal limit as $\varepsilon\to 0$ will yield the hydrodynamic system \eqref{SOH_model} and the formulas for the parameters $c_1(d), c_2(d)$ and $\gamma$.
\end{proof}

\section{The splitting scheme of solving the relaxation model}
\label{app_split}
We start from Eq. \eqref{SOH_relax}. Dropping the superscript $\eta$ for simplicity, we implement the splitting scheme in two steps.
\begin{enumerate}
\item[] Step 1. Solve the conservative part:
\begin{subequations}\label{split_1}
\begin{numcases}{}
\partial_t \rho + \nabla_{\boldsymbol{x}}\cdot(c_1v(\rho)\rho\mathbf\Omega) = 0,\\
\partial_t(\rho\mathbf\Omega) + \nabla_{\boldsymbol{x}}\cdot(c_2v(\rho)\rho\mathbf\Omega\otimes\mathbf\Omega) 
+ d\nabla_{\boldsymbol{x}}(v(\rho)\rho) - \gamma\Delta_{\boldsymbol{x}}(\rho\mathbf\Omega) = 0.
\end{numcases}
\end{subequations}

\item[] Step 2. Solve the relaxation part:
\begin{subequations}\label{split_2}
\begin{numcases}{}
\partial_t \rho = 0, \\
\partial_t(\rho\mathbf\Omega) = \frac{\rho}{\eta}(1-|\mathbf\Omega|^2)\mathbf\Omega.
\end{numcases}
\end{subequations}
\end{enumerate}

Introduce two functions $p$ and $q$ such that $p = \rho\mathbf\Omega_x$ and $q = \rho\mathbf\Omega_y$
where $\mathbf\Omega_x$ and $\mathbf\Omega_y$ are the two components of $\mathbf\Omega$.
The system (\ref{split_1}) can be written in vector form:
\begin{equation}\label{eq_Q}
\partial_t Q + \partial_x(F(Q, \partial_x Q)) + \partial_y(G(Q, \partial_y Q)) = 0,
\end{equation}
where
\begin{align*}
Q = \left(\begin{array}{c} \rho \\ p \\ q \end{array}\right),\quad
F(Q, \partial_x Q) = \left(\begin{array}{c}
c_1v(\rho)p \\ c_2\frac{v(\rho)}{\rho}p^2 + dv(\rho)\rho - \gamma\partial_xp\\
c_2\frac{v(\rho)}{\rho}pq - \gamma\partial_xq\end{array}\right),\\
G(Q, \partial_y Q) = \left(\begin{array}{c}
c_1v(\rho)q\\c_2\frac{v(\rho)}{\rho}pq - \gamma\partial_yp\\
c_2\frac{v(\rho)}{\rho}q^2 + dv(\rho)\rho - \gamma\partial_yq\end{array}\right).
\end{align*}

The explicit time discretization for Eq. \eqref{eq_Q} is given by
\begin{align*}
Q_{i,j}^* = Q_{i,j}^n - \frac{\Delta t}{\Delta x}\left(F_{i+\frac12,j}^n - F_{i-\frac12,j}^n\right)
 - \frac{\Delta t}{\Delta y}\left(G_{i,j+\frac12}^n - G_{i,j+\frac12}^n\right),
\end{align*}
where the numerical flux $F_{i+\frac12,j}$ is defined as
\begin{align*}
F_{i+\frac12,j} = \frac{F(Q_{i,j}) + F(Q_{i+1,j})}{2}
-\frac12 P^2(\frac{\partial F}{\partial Q}(\bar{Q}_{i,j},\overline{\partial_xQ}_{i,j}))(Q_{i+1,j} - Q_{i,j})
\end{align*}
with
\[
\bar{Q}_{i,j} = \frac{Q_{i,j} + Q_{i+1,j}}{2}, \quad
\partial_xQ_{i,j} = \frac{Q_{i+1,j} - Q_{i,j}}{\Delta x}, \quad
\overline{\partial_xQ}_{i,j} = \frac{\partial_xQ_{i,j} + \partial_xQ_{i+1,j}}{2}.
\]
$G_{i,j+\frac12}$ is defined in the similar manner. 
Here $P^2(\frac{\partial F}{\partial Q})$ is a second degree polynomial of a matrix at the intermediate state of $(Q_{i,j},\partial_xQ_{i,j})$ and $(Q_{i+1,j},\partial_xQ_{i+1,j})$; see \cite{Degond_etal_1999} for more details.

\section{The Discrete Fourier Transform for the viscous system \eqref{sys_FT}}
\label{viscous_DFT}
Let $N_x\times N_y$ denote the mesh size over the domain $[0,L_x]\times[0,L_y]$ and $(\Delta x, \Delta y)$ the uniform mesh spacing. 
Each nonoverlapping computational cell is centered at $(x_j, y_k) = (\left(j-\frac12\right)\Delta x, \left(k-\frac12\right)\Delta y)$ for $1\leq j\leq N_x, 1\leq k \leq N_y$.
We study the spatial variable $x$ only and apply the Discrete Fourier Transform on $(\rho_\sigma(x_j,t), \theta_\sigma(x_j,t))$:
\begin{align*}
&\rho_\sigma(x_j,t) = \frac{1}{N_x}\sum_{\xi=0}^{N_x-1} \hat\rho_\sigma(\xi,t)e^{i\frac{2\pi\xi(j-1)}{N_x}}
 = \frac{1}{N_x}\sum_{\xi=0}^{N_x-1} \hat\rho_\sigma(\xi,t)e^{i\frac{2\pi\xi}{L_x}\left(x_j - \frac{\Delta x}{2}\right)},\\
&\theta_\sigma(x_j,t) = \frac{1}{N_x}\sum_{\xi=0}^{N_x-1} \hat\theta_\sigma(\xi,t)e^{i\frac{2\pi\xi(j-1)}{N_x}}
 = \frac{1}{N_x}\sum_{\xi=0}^{N_x-1} \hat\theta_\sigma(\xi,t)e^{i\frac{2\pi\xi}{L_x}\left(x_j - \frac{\Delta x}{2}\right)},
\end{align*}
where
\[
\hat{\rho}_\sigma(\xi,t) = \sum_{j=1}^{N_x}\rho_\sigma(x_j,t)e^{-i\frac{2\pi(j-1)\xi}{N_x}}, \qquad
\hat{\theta}_\sigma(\xi,t) = \sum_{j=1}^{N_x}\theta_\sigma(x_j,t)e^{-i\frac{2\pi(j-1)\xi}{N_x}}.
\]
It follows that
\begin{align*}
\partial_t\left(\begin{array}{c}\hat\rho_\sigma\\\hat\theta_\sigma\end{array}\right)
+i\frac{2\pi\xi}{L_x} A \left(\begin{array}{c}\hat\rho_\sigma\\\hat\theta_\sigma\end{array}\right) = \boldsymbol{0},
\end{align*}
where the matrix
\begin{align*}
A = \left(\begin{array}{cc}
c_1\tilde{v}'(\rho_s)\cos\theta_s & -c_1\tilde{v}(\rho_s)\sin\theta_s \\
-d\frac{\tilde{v}'(\rho_s)}{\rho_s}\sin\theta_s &
- i\frac{2\pi\xi\gamma}{L_x} + c_2\frac{\tilde{v}(\rho_s)}{\rho_s}\cos\theta_s
\end{array}\right).
\end{align*}
Let $\lambda = \mu + i\nu$ be the eigenvalue of $A$. Solving $|A - \lambda {\rm Id}| = 0$ gives
\begin{align*}
\lambda 
&=\frac12\left[\Big(c_1\tilde{v}'(\rho_s) + c_2\frac{\tilde{v}(\rho_s)}{\rho_s}\Big)\cos\theta_s \pm {\rm Re}\sqrt{\Delta}
+ i\Big(\pm{\rm Im}\sqrt{\Delta} - \frac{2\pi\xi\gamma}{L_x}\Big)\right],
\end{align*}
where $\sqrt{\Delta}$ denote the square root of the discriminant, the complex number $\Delta$:
\begin{align*}
\Delta = \left(\Big(c_1\tilde{v}'(\rho_s)- c_2\frac{\tilde{v}(\rho_s)}{\rho_s}\Big)\cos\theta_s
+ i\frac{2\pi\xi\gamma}{L_x}
\right)^2 
+ 4c_1d \frac{\tilde{v}(\rho_s)\tilde{v}'(\rho_s)}{\rho_s}\sin^2\theta_s. 
\end{align*}
\end{appendices}


\end{document}